\newif\ifsubmit

\submittrue

\ifsubmit %
\documentclass[11pt]{article} %
\else %
\documentclass[oneside]{article} %
\fi

\usepackage{article} %
\usepackage{math} %
\usepackage{algo} %
\usepackage{wrapfig} %
\usepackage{placeins} %
\usepackage{graphicx} %
\usepackage[longnamesfirst,numbers]{natbib} %
\usepackage{setspace} %

\usepackage[font={small},labelfont=bf
]{caption}

\ifsubmit \usepackage[letterpaper]{geometry}
\else %

\fi

\allowdisplaybreaks

\everymath{\displaystyle}

\providecommand{\groundset}{\mathcal{N}}     %
\providecommand{\groundsets}{2^{\groundset}} %
\providecommand{\groundcube}{[0,1]^{\groundset}}
\providecommand{\weight}{w\optsub} %
\providecommand{\varweight}{\omega\optpar} %
\newcommand{\defterm}[1]{\emph{#1}}

\providecommand{\Qout}{\hat{Q}} %

\providecommand{\infnorm}[1]{\norm{#1}_{\infty}} %
\providecommand{\apxdelta}{\tilde{\delta}} %

\newcommand{\xout}{\hat{x}}

\NewDocumentCommand{\mlf}{O{F} g g}{
  #1\IfNoValueF{#3}{_{#3}}\IfNoValueF{#2}{\parof{#2}} %
}                                                                %
\NewDocumentCommand{\dmlf}{o g}{%
  \IfNoValueTF{#2}{%
    F'%
    \IfNoValueF{#1}{%
      _{#1}%
    }%
  }{%
    \IfNoValueTF{#1}{%
      F'\parof{#2}%
    }{%
      \parof{F'\parof{#2}}_{#1}%
    }%
  }%
}
\NewDocumentCommand{\ddmlf}{o o g}{%
  \IfNoValueF{#1}{%
    \mleft(%
  }%
  F''
  \IfNoValueF{#3}{\parof{#3}}
  \IfNoValueF{#1}{%
    \mright)%
    \IfNoValueTF{#2}{
      _{#1}                     %
    }{                          %
      _{#1,#2}                  %
    }                           %
  }                             %
}

\begin{document}

\title{Submodular~Function~Maximization in~Parallel
  via~the~Multilinear~Relaxation\footnote{This work is partially
    supported by NSF grant CCF-1526799. University of Illinois,
    Urbana-Champaign, IL 61801. {\tt
      \{chekuri,quanrud2\}@illinois.edu}.}
} %
\author{Chandra Chekuri \and Kent Quanrud}
\maketitle

\begin{abstract}
  \citet{bs-18} recently initiated the study of adaptivity (or
  parallelism) for constrained submodular function maximization, and
  studied the setting of a cardinality constraint. Subsequent
  improvements for this problem by \citet{brs-18} and \citet{en-18}
  resulted in a near-optimal $(1-1/e-\eps)$-approximation in
  $O(\log n/\eps^2)$ rounds of adaptivity. Partly motivated by the
  goal of extending these results to more general constraints, we
  describe parallel algorithms for approximately maximizing the
  multilinear relaxation of a monotone submodular function subject to
  packing constraints. Formally our problem is to maximize $\mlf{x}$
  over $x \in [0,1]^{n}$ subject to $Ax \le \ones$ where $\mlf$ is the
  multilinear relaxation of a monotone submodular function.  Our
  algorithm achieves a near-optimal $(1-1/e-\eps)$-approximation in
  $O(\log^2 m \log n/\eps^4)$ rounds where $n$ is the cardinality of
  the ground set and $m$ is the number of packing constraints. For
  many constraints of interest, the resulting fractional solution can
  be rounded via known randomized rounding schemes that are
  \emph{oblivious} to the specific submodular function. We thus derive
  randomized algorithms with poly-logarithmic adaptivity for a number
  of constraints including partition and laminar matroids, matchings,
  knapsack constraints, and their intersections.

  Our algorithm takes a continuous view point and combines several
  ideas ranging from the continuous greedy algorithm of
  \cite{v-08,ccpv}, its adaptation to the MWU framework for packing
  constraints \cite{cjv-15}, and parallel algorithms for packing LPs
  \cite{ln-93,young-01}. For the basic setting of cardinality constraints,
  this viewpoint gives rise to an alternative, simple to understand
  algorithm that matches recent results \cite{brs-18,en-18}. Our
  algorithm to solve the multilinear relaxation is deterministic if it
  is given access to a value oracle for the multilinear extension and
  its gradient; this is possible in some interesting cases such as the
  coverage function of an explicitly given set system.
\end{abstract}

\thispagestyle{empty}

\newpage
\setcounter{page}{1}

\section{Introduction}
\labelsection{intro}
A real-valued set function $f: \groundsets \rightarrow \mathbb{R}$ over a
finite ground set $\groundset$ is \emph{submodular} iff
\begin{align*}
  f(A) + f(B) \ge f(A \cup B) + f(A \cap B)
  \text{ for all } A, B \subseteq
  \groundset.
  \labelthisequation{submodularity}
\end{align*}
Submodular set functions play a significant role in classical
combinatorial optimization. More recently, due to theoretical
developments and a plethora of applications ranging from algorithmic
game theory, machine learning, and information retrieval \& analysis,
their study has seen a resurgence of interest. In this paper we are
interested in constrained submodular function
\emph{maximization}. Given a non-negative submodular set function
$f:\groundsets \rightarrow \mathbb{R}_+$ over a finite ground set
$\groundset$ the goal is to find $\max_{S \in \mathcal{I}} f(S)$ where
$\mathcal{I}$ is down-closed family of sets that captures some packing
constraint of interest.  The canonical problem here is the cardinality
constrained problem $\max_{|S| \le k} f(S)$. Among many other
applications, this problem captures NP-Hard problems including the
Maximum $k$-Cover problem which can not be approximated to better than
a $(1-1/e-\eps)$-factor for any $\eps > 0$ unless $P=NP$
\cite{f-98}. The cardinality constrained problem has been well-studied
from the 70's with an optimal $(1-1/e)$-approximation established via
a simple greedy algorithm when $f$ is monotone \cite{nwf-78}. There
has been extensive theoretical work in the last decade on
approximation algorithms for submodular function maximization. Several
new algorithmic ideas were developed to obtain improved approximation
ratios for various constraints, and to handle non-monotone functions.
One of these new ingredients is the multilinear relaxation approach
\cite{ccpv} that brought powerful continuous optimization techniques
to submodular function maximization. We refer the reader to a recent
survey \cite{bf-survey} for some pointers to the new developments on
greedy and continuous methods, and to \cite{bz-17} on local search
methods.

Recent applications of submodular function maximization to large data
sets, and technological trends, have motivated new directions of
research. These include the study of faster algorithms in the
sequential model of computation
\cite{bv-14,mbkvk-15,cjv-15,mbk-16,fhk-17}, algorithms in distributed
setting \cite{mksk-13,kmvv-15,mz-15,benw-15,benw-16,lv}, and
algorithms in the streaming setting
\cite{bmkk-14,ck-15,cgq-15}. \citet{benw-16} developed a general
technique to obtain a constant round algorithm in the MapReduce model
of computation that gets arbitrarily close to the approximation
achievable in the sequential setting. The MapReduce model captures the
distributed nature of data but allows for a polynomial amount of
sequential work on each machine. In some very recent work
\citet{bs-18} suggested the study of \emph{adaptivity} requirements
for submodular function maximization which is closer in spirit to the
traditional \emph{parallel} computation model such as the PRAM. To a
first order approximation the question is the following. Assuming that
the submodular function $f$ can be evaluated efficiently in parallel,
how fast can constrained submodular function maximization be done in
parallel? To avoid low-level considerations of the precise model of
parallel computation, one can focus on the number of adaptive rounds
needed to solve the constrained optimization problem; this corresponds
to the depth in parallel computation.  The formal definition of the
notion of adaptivity from \cite{bs-18} is the following.  An algorithm
with oracle access to a submodular function
$f: \groundset \rightarrow \reals$ is $r$-adaptive for an integer $r$
if for $i \in [r]$, every query $q$ to $f$ in round $i$ depends only
on the answers to queries in rounds $1$ to $(i-1)$ (and is independent
of all other queries in rounds $i$ and greater). We believe that the
definition is intuitive and use other terms such as depth, rounds and
iterations depending on the context.

\citet{bs-18} considered the basic cardinality constrained problem and
showed that in the value oracle model (where one assumes black box
access to $f$), one needs $\Omega(\log n/\log \log n)$ rounds of
adaptivity for a constant factor approximation. They also developed a
randomized algorithm with an approximation ratio of $1/3 - \eps$.  In
recent work, \citet{brs-18} and \citet{en-18} described randomized
algorithms that achieved a near-optimal approximation ratio of
$(1-1/e-\eps)$ with $O(\log n/\eps^2)$ adaptivity. The algorithm of
\citet{en-18} uses $\apxO{n \poly{1/\eps}}$ function calls, while the
algorithm of \citet{brs-18} uses $\apxO{n k^2 \poly{1/\eps}}$ function
calls\footnote{We use $\apxO{}$ notation to suppress poly-logarithmic
  factors.}.

We refer the reader to \cite{bs-18} for extensive justification for
the study of adaptivity of submodular function maximization. We
believe that the close connections to parallel algorithms is already a
theoretically compelling motivation. For instance, specific problems
such as Set Cover and Maximum $k$-Cover have been well-studied in the
PRAM model (see \cite{bpt-11} and references therein).  Our goals
here are twofold. First, can we obtain parallel algorithms for other
and more general classes of constraints than the cardinality
constraint? Second, is there a unified framework that cleanly
isolates the techniques and ideas that lead to parallelization for
submodular maximization problems?

\paragraph{Our Contribution:} We address our goals by considering the
following general problem. Given a monotone submodular function
$f:\groundsets \rightarrow \mathbb{R}_+$ maximize $f$ subject to a set of
explicitly given packing constraints in the form
$A x \le \ones, x \in \{0,1\}^n$; here $n = |\groundset|$, and
$A \in [0,1]^{m \times n}$ is a non-negative matrix. Packing
constraints in this form capture many constraints of interest
including cardinality, partition and laminar matroids, matchings in
graphs and hypergraphs, independent sets in graphs, multiple knapsack
constraints, and their intersections to name a few.  To solve these in a
unified fashion we consider the problem of solving in parallel the following
multilinear relaxation:
\begin{align*}
  \text{maximize } \mlf{x}                         %
  \text{ s.t\ }  Ax \le \ones    %
  \text{ and } x \in [0,1]^\groundset. %
  \labelthisequation[\textsf{Pack-ML}]{pack-multi} %
\end{align*}
Here $\mlf: \groundcube \rightarrow \mathbb{R}_+$ is the multilinear
extension of $f$ \cite{ccpv}, a continuous extension of $f$ defined
formally in \refsection{background}. We mention that solving a packing
LP of the form
\begin{align*}
  \text{maximize } \rip{c}{x} \text{~s.t~} Ax \le \ones \text{ and } x \in [0,1]^n
  \labelthisequation[\textsf{Pack-LP}]{pack-lp}
\end{align*}
with $c \ge \zeroes$ is a special case of our problem.

The multilinear relaxation is used primarily for the sake of discrete
optimization. For this reason we make the following convenient
assumption: for every element $j$ of the ground set $\groundset$, the
singleton element $\{j\}$ satisfies the packing constraints, that is,
$Ae_j \le \ones$. Any element which does not satisfy the assumption
can be removed from consideration. We make this assumption for the
rest of the paper as it helps the exposition and avoids uninteresting
technicalities.

Our main result is the following theorem.

\begin{theorem}
  \labeltheorem{main-intro} There is a parallel/adaptive algorithm
  that solves the multilinear relaxation of a monotone submodular
  function subject to $m$ packing constraints with the following
  properties. For a given parameter $\eps > 0$:
  \begin{itemize}
  \item It outputs a $(1-1/e-\eps)$-approximation to the multilinear relaxation.
  \item It runs in $\bigO{\frac{1}{\eps^4}\log^2 m \log n}$ adaptive rounds.
  \item The algorithm is deterministic if given value oracle access to
    $F$ and its gradient $F'$. The total number of oracle queries
    to $F$ and $F'$ is $O(n \poly{\log n/\eps})$.
  \item If only given access to a value oracle for $f$ the algorithm
    is randomized and outputs a $(1-1/e-\eps)$-approximate feasible
    solution with high probability, and deterministically finishes in
    the prescribed number of rounds. The total number of oracle
    accesses to $f$ is $O(n^2 \poly{\log n/\eps})$.
  \end{itemize}
\end{theorem}

Our algorithm solves the continuous relaxation and outputs a
fractional solution $x$. To obtain an integer solution we need to
round $x$. Several powerful and general rounding strategies have been
developed over the years including pipage rounding, swap rounding, and
contention resolution schemes
\cite{ccpv,BansalKNS,cvz-14,cvz-swap-rounding,kst-13,feldman-thesis,bf-survey}.
These establish constant factor integrality gaps for the multilinear
relaxation for many constraints of interest. In particular, for
cardinality constraints and more generally matroid constraints there
is no loss from rounding the multilinear relaxation. Thus solving the
multilinear relaxation in \reftheorem{main-intro} already gives an
estimate of the value of the integer optimum solution. One interesting
aspect of several of these rounding algorithms is the following: with
randomization, they can be made \emph{oblivious} to the objective
function $f$ (especially for monotone submodular functions).  Thus one
can convert the fractional solution into an integer solution without
any additional rounds of adaptivity. Of course, in a fine-grained
parallel model of computation such as the PRAM, it is important to
consider the parallel complexity of the rounding algorithms. This will
depend on the constraint family. We mention that the case of partition
matroids is relatively straight forward and one can derive a
randomized parallel algorithm with an approximation ratio of
$(1-1/e-\eps)$ with poly-logarithmic depth. In \refsection{rounding}
we briefly discuss some rounding schemes that can be easily parallelized.

For the case of cardinality constraint we are able to derive a more
oracle-efficient algorithm with similar parameters as the ones in
\cite{brs-18,en-18}. The efficient version is presented as a
discretization of the continuous algorithm, and we believe it provides
a different perspective from previous work\footnote{\citet[Section
  D]{bs-18} describe very briefly a connection between their
  $1/3$-approximation algorithm and the multilinear relaxation but not
  many details are provided.}. The algorithm can be extended to a
single knapsack constraint while maintaining a depth of
$O(\log {n}/\eps^2)$.

\begin{remark}
  Our parallel algorithm for the multilinear relaxation relies only on
  ``monotone concavity'' of the multilinear extension (as defined in
  \refsection{background}). Thus our parallel alogirthm also applies
  to yield a $(1-1/e-\eps)$-approximation for maximizing any monotone
  concave function subjecting to packing constraints. Even for
  non-decreasing concave functions, which can be optimized almost
  exactly in the sequential setting, it is not clear that they can be
  solved efficiently and near optimally in the parallel setting when
  in the oracle model with black box access to the
  the function and its gradient.
\end{remark}

\begin{remark}
  A number of recent papers have addressed adaptive and parallel
  algorithms for submodular function maximization. Our work was
  inspired by \cite{bs-18,brs-18,en-18} which addressed the
  cardinality constraint. Other independent papers optimized the
  adaptivity and query complexity \cite{fmz-18-monotone}, and obtained
  constant factor approximation for nonnegative nonmonotone functions
  under a cardinality constraint \cite{bbs-18,fmz-18-nn}.  Partly
  inspired by our work, \citet{env-18} obtained improved results for
  approximating the multilinear relaxation with packing constraints.
  First, they obtain a $(1-1/e-\eps)$-approximation for the monotone
  case in $O( \log{n/\eps} \log{m} /\eps^2)$ rounds of
  adaptivity. Second, they are able to handle nonnegative functions
  and obtain a $(1/e - \eps)$-approximation.
\end{remark}

\subsection{Technical overview}
We build upon several ingredients that have been developed in the
past. These include the continuous greedy algorithm for approximating
the multilinear relaxation \cite{v-08,ccpv} and its adaptation to the
multiplicative weight update method for packing constraints
\cite{cjv-15}. The parallelization is inspired by fast parallel
approximation schemes for positive LPs pioneered by \citet{ln-93} and
subsequently developed by \citet{young-01}. Here we briefly sketch the
high-level ideas which are in some sense not that complex.

We will first consider the setting of a single constraint ($m=1$),
which corresponds to a knapsack constraint of the form
$\rip{a}{x} \le 1$. For linear objective functions
$f(x) = \rip{c}{x}$, we know that the optimal solution is obtained by
greedily sorting the coordinates in decreasing order of $c_j/a_j$ and
choosing each coordinate in turn to its fullest extent of the upper
bound $1$ until the budget of one unit is exhausted (the last job may
be fractionally chosen).  One way to parallelize the greedy algorithm
(and taking a continuous view point) while losing only a
$\epsless$-factor is the following. We bucket the ratios $c_j/a_j$
into a logarithmic number of classes by some appropriate
discretization. Starting with the highest ratio class, instead of
choosing only one coordinate, we choose all coordinates in the same
bucket and increase them simultaneously in parallel until the budget
is met or all coordinates reach their upper bound. If the budget
remains we move on to the next bucket. It is not hard to to see that
this leads to a parallel algorithm with poly-logarithmic depth; the
approximation loss is essentially due to bucketing.

Consider now the nonlinear case, \refequation{pack-multi} under a
knapsack constraint.  In the sequential setting, the continuous greedy
algorithm \cite{v-08,ccpv} is essentially the following greedy
algorithm presented as a continuous process over time. At any time
$t$, if $x(t)$ is the current solution, we increase only $x_j$ for the
best ``bang-for-buck'' coordinate $j = \argmax_h \dmlf[h]{x}/a_h$;
here $\dmlf[h]{x}$ is the $h$th coordinate of gradient of the $F$ at
$x$. In the special case of the cardinality constraint, this is the
coordinate with the largest partial derivative.  Multilinearity of $F$
implies that we should increase the same coordinate $j$ until it
reaches its upper bound. A natural strategy to parallelize this greedy
approach is to bucket the ratios of the coordinates (by some
appropriate discretization) and simultaneously increase all
coordinates in the best ratio bucket.  This won't quite work because
$F$ is non-linear and the gradient values decrease as $x$
increases\footnote{This tension is also central to the recent works
  \cite{bs-18,brs-18,en-18}. We believe that it is easier to
  understand it in the continuous setting where one can treat the
  matter deterministically.}. Here is a simple but key idea. Let
$\lambda$ be the current highest ratio and let us call any coordinate
$j$ in the highest bucket a good coordinate.  Suppose we increase all
good coordinates by some $\delta$ until the average ratio of the good
coordinates falls, \emph{after} the increase, to $(1-\eps)
\lambda$. During the step we have a good rate of progress, but the
step size $\delta$ may be very small. However, one can argue that
after the step, the \emph{number} of good coordinates for current
gradient level falls by an $\eps$ fraction. Hence we cannot make too
many such steps this bucket empties, and have made ``dual'' progress
in terms of decreasing the $\ell_{\infty}$-norm of the
gradient. 
This simple scheme suffices to recover a polylogarithmic depth
algorithm for the knapsack constraint. With some additional tricks we
can convert the algorithm into a randomized discrete algorithm that
recovers the parameters of \cite{brs-18,en-18} for the cardinality
constraint. We note that viewing the problem from a continuous point
of view allows for a clean and deterministic algorithm (assuming value
oracles for $F$ and its gradient $F'$).

The more technical aspect of our work is when $m > 1$; that is, when
there are several constraints. Here we rely on a Lagrangean relaxation
approach based on the multiplicative weight update (MWU) method for
positive LPs, which has a long history in theoretical computer science
\cite{ahk-mwu-survey}.  The MWU approach maintains non-negative
weights $w_1,w_2,\ldots,w_m$ on the constraints and solves a sequence
of Lagrangean relaxations of the original problem while updating the
weights. Each relaxed problem is obtained by collapsing the $m$
constraints into a single constraint
$\ripover{A}{\weight}{x} \le \rip{w}{\ones}$ obtained by taking a
weighted linear combination of the original constraints. Note that
this single constraint is basically a knapsack constraint. However,
the weights are updated after each step and hence the knapsack
constraint evolves \emph{dynamically}.  Nevertheless, the basic idea
of updating many variables with the same effective ratio that we
outlined for the single knapsack constraint can be generalized.  One
critical feature is that the weights increase monotonically.  In the
sequential setting, \cite{cjv-15} developed a framework for
\refequation{pack-multi} that allowed a clean combination of two
aspects: (a) an analysis of the continuous greedy algorithm for
proving a $(1-1/e)$-approximation for the multilinear relaxation and
(b) the analysis of the step size and weight updates in MWU which
allows one to argue that the final solution (approximately) satisfies
the constraints.  We borrow the essence of this framework, but in
order to parallelize the algorithm we need both the dual
gradient-decreasing viewpoint discussed above and another idea from
previous work on parallel algorithms for positive LPs
\cite{ln-93,young-01}. Recall that in the setting of a single knapsack
constraint, when we update multiple variables, there are two
bottlenecks for the step size: the total budget and the change in
gradient.  In the MWU setting, the step size is further controlled by
weight update considerations. Accordingly, the step size update rule
is constrained such that if we are increasing along the $j$ coordinate
with a current value of $x_j$, then the updated value is at most
$(1 + \eps^2/\log m)x_j$. This limit is conservative enough to ensure
the weights do not grow too fast, but can only limit the step size a
small number of times before the geometrically increasing coordinates
exceed a certain upper bound.

\paragraph{Organization:} The rest of the paper is organized as
follows. \refsection{background} describes relevant background on
submodular functions and the multilinear extension. In
\refsection{cardinality}, we first describe and analyze an algorithm
for the multilinear relaxation when we have a single cardinality
constraint.  This give an algorithm with $O(\log n/\eps^2)$ depth
assuming oracle access to the multilinear extension $\mlf$ and its
derivative $\dmlf$, which in turn can be implemented via (many more)
oracle calls to $f$ without increasing the adaptivity.  We describe
and analyze our algorithm for general packing constraints in
\refsection{mwu}.  In \refappendix{rpg}, we analyze a randomized
discretization of the continuous algorithm for cardinality constraints
with a better oracle complexity w/r/t $f$.  In \refappendix{knapsack},
we describe and analyze $\bigO{\log n / \eps^2}$-adaptive algorithms
for maximizing a monotone submodular function subject to a single
knapsack constraint.

Note that \refsection{cardinality} is largely included to develop some
intuition ahead of the more complicated constraints in
\refsection{mwu}, but none of the formal observations in
\refsection{cardinality} are invoked explicitly in
\refsection{mwu}. Moreover, the bounds obtained in
\refsection{cardinality} for the cardinality constraint are already
known \citep{brs-18,en-18}. The reader primarily interested in the
main result regarding general packing constraints may prefer to skip
ahead to \refsection{mwu}.


\section{Submodular set functions and the Multilinear relaxation}
\labelsection{background} In this section we provide some relevant
background and notation that we use in the rest of the paper.  Let
$f: \groundsets \to \reals$ assign real values to subsets of
$\groundset$. $f$ is \defterm{nonnegative} if $f(S) \geq 0$ for all
$S \subseteq \groundset$.  $f$ is \defterm{monotone} if
$S \subseteq T$ implies $f(S) \le f(T)$. $f$ is \defterm{normalized}
if $f(\emptyset) = 0$.

We have already seen one definition of submodularity in
\refequation{submodularity}. Another useful (and equivalent)
definition is via \emph{marginal} values. For a real-valued set
function $f:\groundsets \rightarrow \mathbb{R}$, the marginal value of
a set $U$ with respect to a set $S$ is defined as
$f(S \cup U) - f(S)$, which we abbreviate by $f_S(U)$. If $U$ is a
singleton $\{i\}$ we write $f_S(i)$ instead of $f_S(\{i\})$.  We also
use the notation $S+i$ and $S+i+j$ as short hand for $S \cup \{i\}$
and $S \cup \{i,j\}$.  A set function $f$ is submodular iff it
satisfies the following property modeling decreasing marginal returns:
\begin{align*}
  f_S(i) \ge f_T(i) \text{ for all } S \subset T \subseteq \groundset
  \text{ and } i \not \in T.
\end{align*}
The following seemingly restricted form of this property also suffices:
$f_S(i) \ge f_{S+j}(i) \quad \forall S, i,j \not \in S$ and we will
see a continuous analogue of this latter property subsequently.  In
this paper we restrict attention to normalized, nonnegative and
monotone submodular set functions.

\subsection{Multilinear extension and relaxation}

In this section, we outline basic properties of a continuous extension
of submodular functions to the fractional values in $\groundcube$
called the \emph{multilinear extension} \cite{ccpv}.

\begin{notation}
  For two vectors $x,y$, let $x \lor y$ be the coordinatewise maximum
  of $x$ and $y$, and let $x \land y$ denote the coordinatewise
  minimum, and let $x \setminus y = x - x \land y$. We identify an
  element $j$ with the coordinate vector $e_j$, and a set of elements
  $S \subseteq \groundset$ with the sum of coordinate vectors,
  \begin{math}
    \sum_{j \in S} e_j.
  \end{math}
  In particular, for a vector
  \begin{math}
    x \in \groundcube
  \end{math}
  and a set of coordinates $S$, $x \land S$ is the vector obtained
  from $x$ by setting all coefficients not indexed by $S$ to 0, and
  \begin{math}
    x \setminus S = x - x \land S
  \end{math}
  is the vector obtained from $x$
  setting all coordinates indexed by $S$ to 0.
\end{notation}

\begin{definition}
  Given a set function $f: \groundsets \to \reals$, the
  \defterm{multilinear extension} of $f$, denoted $\mlf$, extends $f$
  to the product space $[0,1]^{\groundset}$ by interpreting each point
  $x \in [0,1]^{\groundset}$ as an independent sample
  $S \subseteq \groundset$ with sampling probabilities given by $x$,
  and taking the expectation of $f(S)$. Equivalently,
  $$\mlf(x) = \sum_{S \subseteq \groundset} \parof{\prod_{i \in S}x_i
  \prod_{i \not \in S} (1-x_i)}.$$
  We extend $\mlf$ to the cone $\nnreals^{\groundset}$ by truncation:
  $\mlf{x} = \mlf{x \land \ones}$. where $x \land \ones$ takes the coordinatewise minimum of $x$ and
  the all-ones vector $\ones$.
\end{definition}
We also write
\begin{math}
  \mlf{x}{y} = \mlf{x \lor y} - \mlf{y}
\end{math}
which generalizes the definition of marginal values to the continuous
setting. We let $F'(x)$ denote the gradient of $F$ at $x$ and
$F''(x)$ denote the Hessian of $F$ at $x$. $F'_i(x)$ denotes
the partial derivative of $F$ with respect to $i$, and
$F''_{i,j}(x)$ denotes the second order partial derivative
with respect to $i$ and $j$. The following lemma captures several
submodularity properties of $F$ that it inherits from $f$. The
properties are paraphrased from \cite{v-08,ccpv} and can be
deriveed from the algebraic formula for $F$ and submodularity
of $f$.

\begin{lemma}
  Let $\mlf$ be the multilinear extension of a set function $f$, and
  $x \in \groundcube$.
  \begin{enumerate}
  \item (\defterm{Multilinearity}) For any $i \in \groundset$,
    \begin{math}
      \mlf{x} = \mlf{x \setminus i} + x_i \mlf{i}{x \setminus i}.
    \end{math}
    In particular, $\mlf{x}$ is linear in $x_i$.
  \item (\defterm{Monotonicity}) For any $i \in \groundset$,
    \begin{math}
      \dmlf[i]{x} =             %
      \mlf{i}{x \setminus i}.
    \end{math}
    In particular, if $f$ is monotone, then $\dmlf$ is nonnegative,
    and $\mlf$ is monotone (that is, $F(y) \ge F(x)$ if $y \ge x$).
  \item For any $i \neq j \in \groundset$, for
    $y = x \setminus \setof{i,j}$, we have
    \begin{align*}
      \ddmlf[i][j]{x} = \mlf{y \lor \setof{i,j}} %
      - \mlf{y \lor i} - \mlf{y \lor j} %
      + %
      \mlf{y}.
    \end{align*}
    If $f$ is submodular, then
    \begin{math}
      \ddmlf[i][j]{x} \leq 0.
    \end{math}
  \item (\defterm{Monotone concavity}) For any $d \geq \zeroes$, the function
    $\lambda \mapsto \mlf{x + \lambda d}$ is concave in $\lambda$
    (whenever $F(x + \lambda d)$ is defined).
  \end{enumerate}
\end{lemma}

\paragraph{Multilinear relaxation:} The multilinear extension $\mlf$
of a submodular function $f$ has many uses, but a primary motivation
is to extend the relax-and-round framework of approximation algorithms
for linear functions to submodular function maximization. Given a
discrete optimization problem of the form
$\max_{S \in \mathcal{I}} f(S)$ we relax it to the continuous
optimization problem $\max_{x \in P_\mathcal{I}} F(x)$ where
$P_{\mathcal{I}}$ is a polyhedral or convex relaxation for the
feasible solutions of constraint family $\mathcal{I}$. The problem
$\max_{x \in P_\mathcal{I}} F(x)$ is referred to as the multilinear
relaxation. It is useful to assume that linear optimization over
$P_\mathcal{I}$ is feasible in polynomial time in which case it is
referred to as \emph{solvable}.  The multilinear relaxation is not
exactly solvable even for the simple cardinality constraint polytope
$\{ x \in [0,1]^n: \sum_{i} x_i \le k\}$. The continuous greedy
algorithm \cite{v-08} gives an optimal $(1-1/e)$ approximation for
solvable polytopes when $f$ is monotone. Our focus in this paper is
the restricted setting of explicit packing constraints.

\paragraph{Preprocessing:}
Recall that we made an assumption that for all $j \in \groundset$,
$A e_j \le \ones$. With this assumption in place we can do some useful
preprocessing of the given instance. First, we can get lower and
upperbounds on $\opt$, the optimum solution value for the
relaxation. We have $\opt \ge \ell = \max_j f(j)$ and
$\opt \le \sum_j f(j) = u \le n \ell$. Since we are aiming for a
$(1-1/e - \eps)$-approximation we can assume that for all $j$,
$f(j) \ge \frac{\eps}{n} \ell$; any element which does not satisfy
this assumption can be discarded and the total loss is at most
$\eps \opt$. Further, we can see, via sub-additivity of $f$ and $F$
that
$\mlf{\frac{\eps}{n}\ones} \le \frac{\eps}{n} \sum_j f(j) \le \eps
\opt$.  We can also assume that $A_{i,j} = 0$ or $A_{i,j} \ge \eps/n$
for all $i,j$; if $A_{i,j} < \eps/n$ we can round it down to $0$. Let
$A'$ be the modified matrix. If $A'x \le \ones$ then we have that
$Ax \le (1+\eps) \ones$. Therefore $A(1-O(\eps))x \le \ones$.  From
monotone concavity we also see that
$\mlf{(1-O(\eps))x} \ge (1-O(\eps))\mlf{x}$. Thus, solving with
respect to $A'$ does not lose more than a $(1-O(\eps)$ multiplicative
factor when compared to solving with $A$.

\paragraph{Evaluating $\mlf$ and $\dmlf$:} The
formula for $\mlf{x}$ gives a natural random sampling algorithm to
evaluate $\mlf{x}$ in expectation. Often we need to evaluate
$\mlf{x}$ and $\dmlf{x}$ to high accuracy. This issue has been
addressed in prior work via standard Chernoff type concentration
inequalities when $f$ is non-negative.

\begin{lemma}[\citealp{cjv-15}]
  \labellemma{dmlf-sample}
  Suppose $\dmlf[i]{x} \in [0,M']$. Then with $r = O(\frac{1}{\eps^2} p
  \log d)$ parallel evaluations of $f$ one can find an estimate
  $Z$ of $\dmlf[i]{x}$ such that $\probof{ |Z - \dmlf[i]{x}| \ge \eps
    \dmlf[i]{x} + \frac{\eps}{p} M'} \le \frac{1}{d^3}$. Similarly, if
  $\mlf{x} \in [0,M]$, then with $r = O(\frac{1}{\eps^2}p
  \log d)$ parallel evaluations of $f$, one can find an estimate
  $Z$ of $\mlf{x}$ such that $\probof{ |Z - \mlf{x}| \ge \eps
    \mlf{x} + \frac{\eps}{p} M} \le \frac{1}{d^3}$.
\end{lemma}

Choosing $d = n$ and $p = n$ we can estimate $\dmlf[i]{x}$ and
$\mlf(x)$ to within a $(1\pm \eps)$ multiplicative error, and an
additive error of $\frac{\eps}{n}M'$ and $\frac{\eps}{n}M$
respectively. Via the preprocessing that we already discussed, we can
assume that $M \le n \opt$ and $M' \le \opt$.  For any $x$ such that
$\mlf{x} \ge \frac{\eps}{n} \opt$ we can set $p = n^2/\eps$ to obtain
a $(1+\eps)$-relative approximation. Similarly if $\dmlf[i]{x} \ge
\frac{\eps}{n} \opt$ we can obtain a $(1+\eps)$-relative approximation
by setting $p = n/\eps$.

In some cases an explicit and simple formula for $F$ exists from which
one can evaluate it deterministically and efficiently. A prominent
example is the coverage function of a set system. Let $f$ be defined
via a set system on $n$ sets $A_1,A_2,\ldots,A_n$ over a universe
$\mathcal{U}$ of size $r$ as follows. For $S \subseteq [n]$ we
let $f(S) = \cup_{i \in S} A_i$, the total number of elements covered
by the sets in $S$. It is then easy to see that
$$\mlf{x} = \sum_{e \in \mathcal{U}} \parof{1 - \prod_{i: e \in A_i} (1-x_i)}.$$
Thus, given an explicit representation of the set system,
$\mlf{x}$ and $\dmlf{x}$ can be evaluated efficiently and
deterministically\footnote{We ignore the numerical issues involved in
  the computation. One can approximate the quantities of interest with a
  small additive and multiplicative error via standard tricks.}.

Throughout the paper we assume that $\eps > 0$ is sufficiently
small. We also assume that $\eps > \poly{1/n}$, since otherwise
sequential algorithms already achieve $1/\eps$-adaptivity.


\section{Parallel maximization with a cardinality constraint}
\labelsection{cardinality}

We first consider the canonical setting of maximizing the multilinear
extension of a submodular function subject to a cardinality constraint
specified by an integer $k$. The mathematical formulation is below.
\begin{align*}
  \text{maximize } \mlf{x} \text{ over } x \in \nnreals^{\groundset}
  \text{ s.t.\ } \rip{\ones}{x} \leq k.
\end{align*}
This problem was already considered and solved to satisfaction by
\citet{brs-18} and \citet{en-18}. The approach given here is different
(and simple enough), and is based on the \algo{continuous-greedy}
algorithm of \citet{ccpv}, specialized to the cardinality constraint
polytope. Establishing this connection lays the foundation for general
constraints in \refsection{mwu}. That said, there is no formal
dependence between \refsection{mwu} and this section. As the bounds
presented in this section have been obtained in previous work
\cite{brs-18,en-18}, the reader primarily interested in new results
may want to skip ahead to \refsection{mwu}.

\begin{figure}[tb]
  \centering
  \begin{minipage}{.5\paperwidth}
    \begin{framed}
      \ttfamily\raggedright
      \underline{parallel-greedy($f$,$\groundset$,$k$,$\eps$)}
      \begin{steps}
      \item $x \gets \zeroes$
      \item $\lambda \gets \opt$ %
        \commentcode{or any upper bound for $\opt$}
      \item while $\rip{x}{\ones} \leq k$ and
        $\lambda \geq e^{-1} \opt$
        \begin{steps}
        \item \labelstep{pg-good-coordinates} let
          \begin{math}
            S = \setof{ %
              j \in \groundset %
              \where %
              F'_j(x) \geq \frac{\epsless \lambda}{k} %
            }%
          \end{math}
        \item \labelstep{pg-inc-loop} while $S$ is not empty and
          $\rip{x}{\ones} \leq k$
          \begin{steps}
          \item chose $\delta$ maximal s.t.\
            \begin{steps}
            \item \labelstep{pg-apx-condition}
              \begin{math}
                \mlf{x + \delta S}{x} \geq \epsless^2 \lambda
                \frac{\delta \sizeof{S}}{k}
              \end{math}
            \item $\rip{x + \delta S}{\ones} \leq k$
            \end{steps}
          \item \labelstep{pg-inc} $x \gets x + \delta S$
          \item \labelstep{pg-update-S} update $S$
          \end{steps}
        \item \labelstep{pg-threshold}
          \begin{math}
            \lambda \gets \epsless \lambda
          \end{math}
        \end{steps}
      \item return $x$
      \end{steps}
      \caption{A parallel implementation of the
        \algo{continuous-greedy} algorithm specialized to the
        cardinality polytope.\labelfigure{parallel-greedy}}
    \end{framed}
  \end{minipage}
\end{figure}

We propose the algorithm \algo{parallel-greedy}, given in
\reffigure{parallel-greedy}. It is a straightforward parallelization
of the original \algo{continuous-greedy} algorithm due to
\citet{ccpv}, specialized to the cardinality polytope.
\algo{continuous-greedy} is an iterative and monotonic algorithm that,
in each iteration, computes the gradient $\dmlf{x}$ and finds the
point $v$ in the constraint polytope that maximizes
$\rip{\dmlf{x}}{v}$.  In the case of the cardinality polytope, $v$ is
$e_j$ for the coordinate $j = \argmax_h \dmlf[h]{x}$ with the largest
gradient.  \algo{continuous-greedy} then adds $\delta e_j$ to $x$ for
a fixed and conservative step size $\delta > 0$. The new algorithm
\algo{parallel-greedy} makes two changes to this algorithm. First,
rather than increase $x$ along the single best coordinate, we identify
all ``good'' coordinates with gradient values nearly as large as the
best coordinate, and increase along all of these coordinates
uniformly.  Second, rather than increase $x$ along these coordinates
by a fixed increment, we choose $\delta$ dynamically. In particularly,
we greedily choose $\delta$ as large as possible such that, after
updating $x$ and thereby decreasing the gradient coordinatewise,
the set of good coordinates is still nearly as good on average.

The dynamic choice of $\delta$ accounts for the fact that increasing
multiple coordinates simultaneously can affect their gradients.  The
importance of greedily choosing the step size is to geometrically
decrease the number of good coordinates. It is shown below (in
\reflemma{pg-depth}) that, when the many good coordinates are no
longer nearly-good on average, a substantial fraction of these
coordinates are no longer good.  When there are no nearly-good
coordinates remaining, the threshold for ``good'' decreases. The
threshold can decrease only so much before we can conclude that the
current solution $x$ cannot be improved substantially and obtains the
desired approximation ratio. Thus \algo{parallel-greedy} takes a
primal-dual approach equally concerned with maximizing the objective
as driving down the gradient.

We first assume oracle access to values $\mlf{x}$ and gradients
$\dmlf{x}$. The algorithm and analysis immediate extends to
approximate oracles that return relative approximation to these
quantities. Such oracles do exist (and are readily parallelizable) for
many real submodular functions of interest. Given oracle access to
$f$, one can implement sufficiently accurate oracles to $\mlf{x}$ and
$\dmlf{x}$ without increasing the depth but with many more oracle
calls to $f$. In \refsection{cardinality-oracle}, we present a
randomized discretization of \algo{parallel-greedy} that improves the
oracle compliexity w/r/t $f$. Note that the algorithms in
\citep{brs-18,en-18} call $f$ directly and do not assume oracle access
to $\mlf$ or $\dmlf$.

\subsection{Approximation ratio}
We first analyze the approximation ratio of the output solution
$\xout$. The main observation is that $\lambda$ is an upper bound on
the gap $\opt - \mlf{x}$.
\begin{lemma}
  \labellemma{pg-threshold}
  At any point, we have $\lambda \geq \opt - \mlf{x}$.
\end{lemma}

\begin{proof}
  The claim holds initially. Whenever $x$ is increased,
  $\opt - \mlf{x}$ decreases since $\mlf{x}$ is monotone, and hence
  the claim continues to hold. Whenever $\lambda$ is about to be
  decreased in \refstep{pg-threshold}, we have $S$ empty (or the
  algorithm terminates since $\rip{x}{\ones} = k$) with respect
  to the current value of $\lambda$. Thus, if $z$ is an optimal
  solution
  then we have
  \begin{align*}
    \opt - \mlf{x}               %
    &\tago{\leq}                     %
      \mlf{z}{x}              %
      \tago{\leq}                    %
      \rip{\dmlf{x}}{z \lor x - x}
      \tago{\leq}
      \rip{\dmlf{x}}{z}
      \tago{\leq}                    %
      \epsless \frac{\lambda}{k} \rip{z}{\ones}
      \tago{\leq}                    %
      \epsless \lambda
  \end{align*}
  by \tagr monotonicity of $F$, \tagr monotonic concavity of $F$,
  \tagr monotonicity of $F$ (implying $\dmlf{x} \geq \zeroes$) and
  $z \lor x - x \leq z$, \tagr emptiness of $S$ w/r/t $\lambda$, and
  \tagr the fact that
  \begin{math}
    \rip{z}{\ones} \leq k.
  \end{math}
\end{proof}
The connection between $\lambda$ and $\opt - \mlf{x}$ allows us to
reinterpret \refstep{pg-apx-condition} as saying that we are closing
the objective gap at a good rate in proportion to the increase in the
(fractional) cardinality of $x$. This is the basic invariant in
standard analyses of the greedy algorithm that implies that greedy
achieves a (near) $\parof{1 - e^{-1}}$-approximation, as follows.
\begin{lemma}
  The output $\xout$ satisfies
  \begin{math}
    \mlf{\xout} \geq \apxless \parof{1 - e^{-1}} \opt.
  \end{math}
\end{lemma}
\begin{proof}
  Let $t = \sum_i x_i$ be the total sum of the coordinates.
  From the preceding lemma and the choice of $\delta$ in the
  algorithm, we have
  \begin{math}
    \mlf{x + \delta S}{x} \geq \epsless^2 \frac{\delta
      \sizeof{S}}{k} \parof{\opt - \mlf{x}}
  \end{math}
  in \refstep{pg-apx-condition}, hence
  \begin{align*}
    \frac{d \mlf{x}}{d t} \geq \frac{\epsless^2}{k}
    \parof{\opt - \mlf{x}},
  \end{align*}
  hence
  \begin{align*}
    \mlf{x} \geq \parof{1 - \exp{-\epsless^2 t / k}}\opt.
  \end{align*}
  In particular, if $t = \rip{\ones}{x} = k$ at the end of the algorithm,
  we have
  \begin{align*}
    \mlf{x} \geq \apxless \parof{1 - e^{-1}} \opt.
  \end{align*}
  If $\lambda \leq e^{-1} \opt$, then
  $\mlf{x} \geq \parof{1 - e^{-1}} \opt$.
  In either case, the output $\xout$ satisfies
  \begin{math}
    \mlf{\xout} \geq \apxless \parof{1 - e^{-1}} \opt.
  \end{math}
\end{proof}

\subsection{Iteration count}
We now analyze the iteration count of \algo{parallel-greedy}. The key
observation lies in line \refstep{pg-apx-condition}. If $\delta$ is
determined by line \refstep{pg-apx-condition}, then the margin of
taking $S$ uniformly has dropped significantly. In this case, as the
next lemma shows, a significant fraction of the coordinates in $S$
must have had their marginal returns decrease enough to force them to
drop out of $S$. The iteration can then be charged to the geometric
decrease in $\sizeof{S}$.

\begin{lemma}
  \labellemma{pg-depth}
  If
  \begin{math}
    \mlf{x + \delta S}{x} = \epsless^2 \lambda \frac{\delta \sizeof{S}}{k},
  \end{math}
  then the step \refstep{pg-update-S} decreases $\sizeof{S}$ by at least a
  $\epsless$-multiplicative factor. This implies that, for fixed
  $\lambda$, the loop at \refstep{pg-inc-loop} iterates at most
  $\bigO{\frac{\log n}{\eps}}$ times, and at most
  $\bigO{\frac{\log n}{\eps^2}}$ times total. That is, each step in
  $\algo{greedy}$ iterates at most $\bigO{\frac{\log n}{\eps^2}}$
  times.
\end{lemma}
\begin{proof}
  Let $x'$ and $S'$ denote the values of $x$ and $S$ before
  updating, and let $x''$ and $S''$ denote the values of $x$ and $S$
  after. We want to show that
  \begin{math}
    \sizeof{S''} \leq \epsless \sizeof{S'}.
  \end{math}
  We have
  \begin{align*}
    \epsless^2 \lambda \frac{\delta \sizeof{S'}}{k}
    &\tago{=}                         %
      \mlf{x''}{x}
      \geq                      %
      \rip{\dmlf{x''}}{\delta S'}
      \tago{\geq}
      \rip{\dmlf{x''}}{\delta S''}
      \tago{\geq} \epsless \lambda \frac{\delta \sizeof{S''}}{k}.
      \labelthisequation{pg-depth-derivation} %
  \end{align*}
  by \tagr choice of $\delta$, \tagr monotonicity, and \tagr
  definition of $S''$.  Dividing both sides by $\epsless \lambda$, we
  have
  \begin{math}
    \sizeof{S''} \leq \epsless \sizeof{S'}.
  \end{math}
\end{proof}

One implementation detail is finding $\delta$ in the inner loop.  We
can assume that $\poly{\eps / n} \leq \delta \leq k$ (since below
$\poly{\eps/n}$, the gradient $\dmlf{x}$ does not change
substantially). It is easy to see that a (say)
$\parof{1 + \eps / 2}$-multiplicative approximation of the exact value
of $\delta$ suffices. (A more detailed discussion of approximating
$\delta$ is in the more subtle setting of generic packing constraints
is given later in \refsection{pmwu-work}). Hence we can try all
$\bigO{\log{n} / \eps}$ powers of $\epsmore$ between $\poly{\eps/n}$
and 1 to find a sufficiently good approximation of $\delta$. A second
implementation detail regards to initial value of $\lambda$ for upper
bounding $\opt$. Standard tricks allow us to obtain a constant factor
without increasing the depth; see the related discussing w/r/t general
packing constraints in \refsection{pmwu-work}.


\subsection{Oracle complexity w/r/t $f$}

\labelsection{cardinality-oracle}

The preceding algorithm and analysis were presented under the
assumption that gradients of the multilinear extension $\mlf$ were
easy to compute (at least, approximately). This assumption holds for
many applications of interest. In this section, we consider a model
where we only have oracle access to the underlying set function $f$.

We first note that $\mlf{x}$ and $\dmlf{x}$ can still be approximated
(to sufficient accuracy) by taking the average of $f(Q)$ for many
random samples $Q \sim x$. To obtain $\epspm$-accuracy with high
probability for either $\mlf{x}$ or a single coordinate of $\dmlf{x}$,
one requires about $\bigO{\frac{n \log n}{\eps^2}}$ samples, each of
which may be computed in parallel (see \reflemma{dmlf-sample}). Thus
\algo{parallel-greedy} still has $\apxO{\frac{1}{\eps^2}}$ depth in
this model. However, the total number of queries to $f$ increases to
$\apxO{n^2 \poly{1/\eps}}$, because computing an entire gradient to
assemble $S$ in line \refstep{pg-good-coordinates} requires
$\apxO{\frac{n^2}{\eps^2}}$ queries to $f$.
\begin{figure}[t]
  \centering
  \begin{minipage}{.5\paperwidth}
    \begin{framed}
      \ttfamily\raggedright
      \underline{randomized-parallel-greedy($f$, $\groundset$, $k$, $\eps$)}
      \begin{steps}
      \item $Q \gets \emptyset$, $t \gets 0$, $\lambda \gets \opt$ %
        \commentcode{or any $\lambda \geq \opt$}
      \item while $t \leq (1 - 2 \eps) k$ and $\lambda \geq e^{-1} \opt$
        \begin{steps}
        \item let
          \begin{math}
            S = \setof{ %
              j \in \groundset %
              \where %
              f_Q(j) \geq \frac{\epsless \lambda}{k} %
            }%
          \end{math}
        \item \labelstep{rpg-inc-loop} while $S$ is not empty and
          $t \leq (1 - 2 \eps) k$
          \begin{steps}
          \item chose $\delta$ maximal s.t.\
            \begin{steps}
            \item \labelstep{pg-apx-condition}
              \begin{math}
                \mlf{Q + \delta S}{Q} \geq \epsless^2 \lambda
                \frac{\delta \sizeof{S}}{k}
              \end{math}
            \item $t + \delta \sizeof{S} \leq (1 - 2\eps) k$
            \end{steps}
          \item \labelstep{rpg-sample} sample $R \sim \delta S$
          \item \labelstep{rpg-inc} $Q \gets Q \cup R$,
            $t \gets t + \delta \sizeof{S}$
          \item update $S$
          \end{steps}
        \item \labelstep{pg-threshold}
          \begin{math}
            \lambda \gets \epsless \lambda
          \end{math}
        \end{steps}
      \item return $Q$
      \end{steps}
      \caption{A randomized, combinatorial variant of the previous
        \algo{parallel-greedy} algorithm for cardinality
        constraints.\labelfigure{randomized-parallel-greedy}}
    \end{framed}
  \end{minipage}
\end{figure}

To reduce the oracle complexity w/r/t $f$, we propose the alternative
algorithm \algo{randomized-parallel-greedy} in
\reffigure{randomized-parallel-greedy}, which is guided by the
previous \algo{parallel-greedy} algorithm, but maintains a discrete
set $Q \subseteq \groundset$ rather than a fractional solution
$x$. The primary difference is in steps \refstep{rpg-sample} and
\refstep{rpg-inc}, where rather than add the fractional solution
$\delta S$ to our solution, we first sample a set $R \sim \delta S$
(where each coordinate in $S$ is drawn independently with probability
$\delta$), and then we add $R$ to the running solution. The primary
benefit to this rounding step is that computing the gradient
$\dmlf{x}$ is replaced by computing the margins $f_Q$, which requires
only a constant number of oracle calls per element.

We defer the analysis of \algo{randomized-parallel-greedy} to
\refappendix{rpg}. At a high level, one can see that the key points to
the analysis of \algo{parallel-greedy} now hold in
expectation. Further techniques from randomized analysis adapt the
essential invariants from \algo{parallel-greedy} to the additional
randomization to obtain the following bounds.

\begin{theorem}
  \labeltheorem{rpg}
  Let $\eps > 0$ be given, let $f: \groundsets \to \nnreals$ be a
  normalized, monotone submodular function in the oracle model, and
  let $k \in \naturalnumbers$. Then with high probability,
  \algo{randomized-parallel-greedy} computes a $\epsless\parof{1 -
    e^{-1}}$ multiplicative approximation to the maximum value set of
  cardinality $k$ with
  \begin{math}
    \bigO{\frac{\log n}{\eps^2}}
  \end{math}
  expected adaptivity and
  \begin{math}
    \apxO{\frac{n}{\eps^4}}
  \end{math}
  expected oracle calls to $f$.
\end{theorem}


\section{Parallel maximization with packing constraints} %
\labelsection{mwu} %
We now consider the general setting of maximizing the multilinear
relaxation in the setting of explicit packing constraints in the form
below.
\begin{align*}
  \text{maximize } \mlf{x} \text{ over } x \in \groundcube
  \text{ s.t.\ } A x \leq \ones.
\end{align*}

\begin{figure}[tb!]
  \begin{framed}
    \ttfamily\raggedright
    \underline{parallel-mwu-greedy}
    \begin{steps}
    \item for each $j \in [n]$
      \begin{steps}
      \item choose $x_j$ maximal in $[0,1]$ s.t.\
        $A_{ij} x_j \leq \frac{\eps}{n}$ for all $i \in [m]$
      \end{steps}
    \item $t \gets \zeroes$, $\lambda \gets \opt$ \commentcode{or any
        upper bound for $\opt$}
    \item define $\weight = \weight(x)$ by
      \begin{math}
        \weight{i} = \exp{\frac{\log m}{\eps} \parof{A x}_i}
      \end{math}
      for $i \in [m]$
    \item \labelstep{pmwu-loop} while $t < 1$ and
      $\lambda \geq e^{-1} \opt$
      \begin{steps}
      \item $W \gets \rip{\weight}{\ones}$
      \item let
        \begin{math}
          S = \setof{j \in \groundset \where
            \frac{\dmlf[j]{x}}{\ripover{A}{\weight}{e_j}} \geq
            \epsless^3 \frac{\lambda}{W} %
            \andcomma %
            \dmlf[j]{x} \geq \frac{\eps \epsless \lambda}{n} }
        \end{math}
      \item
        \labelstep{pmwu-inc-loop}
        while $S$ is not empty, $\epsless \rip{\weight}{\ones}
        \leq W$, and $t \leq 1$
        \begin{steps}
        \item
          \begin{math}
            \gamma \gets \frac{W}{\pripover{A}{\weight}{x \land S}}
          \end{math}
        \item
          \labelstep{pmwu-greedy-step-size}
          choose $\delta > 0$ large as possible s.t.\ for
          $x' = x + \delta \gamma \parof{x \land S}$
          \begin{steps}
          \item \labelstep{pmwu-gradient-condition}
            \begin{math}
              \mlf{x'}{x} \geq \epsless^4 \delta \lambda.
            \end{math}
          \item \labelstep{pmwu-weight-condition}
            \begin{math}
              \gamma \delta \leq \frac{\eps^2}{4 \log m}
            \end{math}
          \item
            \begin{math}
              t + \delta \leq 1.
            \end{math}
          \end{steps}
        \item \labelstep{pmwu-inc}
          $x \gets x + \delta \gamma (x \land S)$ \commentcode{update $x$}
        \item $t \gets t + \delta$ \commentcode{update time}
        \item \labelstep{pmwu-update} update $\weight$ and $S$
        \end{steps}
      \item if $\epsless \rip{\weight}{\ones} \leq W$ and $S$ is empty
        \begin{steps}
        \item \labelstep{pmwu-threshold}
          \begin{math}
            \lambda \gets \epsless \lambda.
          \end{math}
        \end{steps}
      \end{steps}
    \item return $x$
    \end{steps}
    \caption{A parallel implementation of the
      MWU/\algo{continuous-greedy} algorithm of
      \citet{cjv-15}.\labelfigure{pmwu}}
  \end{framed}
\end{figure}

We refer the reader to some preprocessing steps outlined in
\refsection{background}.  In \reffigure{pmwu}, we give a parallel
algorithm that combines the many-coordinate update and greedy step
size of \algo{parallel-greedy} with multiplicative weight update
techniques that navigates the packing constraints. The high-level MWU
framework follows the one from \cite{cjv-15}.

We briefly explain the algorithm. The framework from \cite{cjv-15} has
a notion of time, maintained in the variable $t$, that goes from $0$
to $1$. The algorithm maintains non-negative weights $w_i$ for each
constraint $i$ that reflect how tight is each constraint. In the
sequential setting, the algorithm in \cite{cjv-15} combines
\algo{continuous-greedy} and MWU as follows. In each iteration, given
the current vector $x$, it finds a solution to the following linear
optimization problem with a single non-trival constraint obtained via
a weighted linear combination of the $m$ packing constraints:
\begin{align*}
  \max \rip{\dmlf{x}}{y}        %
  \text{ s.t.\ } \ripover{A}{\weight}{x} \le \rip{\weight}{\ones}          %
  \text{ and } y \ge \zeroes.
\end{align*}
The optimum solution to this relaxation is a single
coordinate solution $\gamma e_j$ where $j$ maximizes the ratio
$\frac{\dmlf[h]{x}}{\rip{w}{Ae_h}}$. The algorithm then updates $x$ by
adding $\delta e_j$ for some appropriately small step size $\delta$
and then updates the weights. The weights are maintained according to
the MWU rule and (approximately) satisfy the invariant that
$w_i = \exp{\eta (Ax)_i)}$ for $\eta = \Theta(\log m/\eps)$.

The parallel version differs from the sequential version as follows.
When solving the Lagrangean relaxation it considers all good
coordinates (the set $S$) whose ratios are close to
$\lambda = \max_h \frac{\dmlf[h]{x}}{\rip{w}{Ae_h}}$ and
\emph{simultaneously} updates them. The step size has to be adjusted
to account for this, and the adjusted step size is a primary
difference from the algorithm in \cite{cjv-15}. The sequential
algorithm takes a greedy step for the sake of obtaining width
independence. In the parallel setting, two different considerations
come in to play. First, the simultaneous update to many coordinates
means that the step size needs to be small enough such that the
gradient does not change too much, but that it does change
sufficiently so that we can use an averaging argument to limit the
number of iterations. Second, if the gradient is not the bottleneck,
then the bottleneck comes from limiting the change in $x$ to ensure
the weights do not grow too rapidly. In this case, we ensure that each
coordinate $j \in S$ increases by at least $(1 + \eps^2/\log m)$
multiplicative factor, which can only happen a limited number of times
due to the starting value of $x$.

We organize the formal analysis into four parts.  The first part,
\refsection{pmwu-packing}, concerns the packing constraints, and shows
that the output $\xout$ satisfies $A \xout \leq \apxmore \ones$.  The
second part, \refsection{pmwu-apx}, concerns the approximation ratio,
and shows that the output $\xout$ has an approximation factor of
$\mlf{\xout} \geq \parof{(1 - e^{-1})-O(\eps)}\opt$. The third part,
\refsection{pmwu-iterations}, analyzes the number of iterations and
shows that each step in \reffigure{pmwu} is executed at most
$\apxO{\frac{1}{\eps^4}}$ times. The last part,
\refsection{pmwu-work}, addresses the total number of oracle
calls. The lemmas in these parts together prove
\reftheorem{main-intro}.

We first note the monotonicity of the various variables at play.
\begin{observation}
  Over the course of the algorithm, $x$ is increasing, $\weight$ is
  increasing, $t$ is increasing, $\mlf{x}$ is increasing, $W$ is
  increasing, $\dmlf{x}$ is decreasing, and $\lambda$ is
  decreasing. Within the loop at \refsubsteps{pmwu-inc-loop}, $S$ is
  decreasing.
\end{observation}

\subsection{Feasibility of the packing constraints}
\labelsection{pmwu-packing} We first show that the algorithm satisfies
the packing constraints to within a $(1+O(\eps))$-factor. The first
fact shows that the weights grow at a controlled rate as $x$
increases. The basic proof idea, which appears in \citet{young-01},
combines the fact that we increase (some coordinates) of $x$ by a
small geometric factor, and the fact that $x$ is recursively
near-feasible. This implies that the increase in load of any
constraint is by at most a small additive factor, hence the weights
(which exponentiate the loads) increase by at most a small geometric
factor.

\begin{lemma}
  \labellemma{pmwu-weight-step} At the beginning of each iteration of
  step \refstep{pmwu-inc}, if $A x \leq 2 \ones$, then
  \begin{align*}
    \weight(x + \delta \gamma (x \land S)) %
    \leq                                   %
    \epsmore \rip{\weight}{\ones}.
  \end{align*}
\end{lemma}

\begin{proof}
  For each constraint $i \in [m]$, we have
  \begin{align*}
    \weight{i}\parof{x + \delta \gamma (x \land S)} %
    &=                                         %
      \weight{i}\parof{{\delta \gamma (x \land S)}} \weight{i}(x)
      \tago{\leq}                        %
      \weight{i}\parof{\frac{\eps^2}{4 \log m} (x \land
      S)}\weight{i}(x)
    \\
    &                                                    %
      =                           %
      e^{\frac{\eps}{4} (A (x \land S))_i} \weight{i}(x) %
      \tago{\leq}                                               %
      e^{\eps / 2} \weight{i}(x)                         %
      \tago{\leq}                                               %
      \epsmore \weight{i}(x)
  \end{align*}
  by \tagr choice of $\delta$ per \refstep{pmwu-weight-condition},
  \tagr $\parof{A (x \land S)}_i \leq 2$, and \tagr upper bounding the
  Taylor expansion of $e^{\eps / 2}$.
\end{proof}

\begin{lemma}
  \labellemma{pmwu-weight-growth}
  At the beginning of each iteration of step
  \refstep{pmwu-inc}, if $A x \leq 2 \ones$, then
  \begin{align*}
    \rip{\weight(x + \delta \gamma (x \land S))}{\ones} %
    \leq                                                %
    \parof{1 + \delta \epsmore \frac{\log m}{\eps}} \rip{\weight}{\ones}. %
    \labelthisequation{pmwu-packing}
  \end{align*}
\end{lemma}

\begin{proof}
  The following is a standard proof from the MWU framework, where the
  important invariant is preserved by choice of $\delta$ w/r/t
  \refstep{pmwu-weight-condition}.  Let
  $x' = x + \delta \gamma \parof{x \land S}$.  Define
  \begin{math}
    \varweight{\tau} %
    = %
    \weight(x + \tau \gamma \parof{x \land S}),
  \end{math}
  where we recall that
  \begin{align*}
    \weight{i}(x + \tau \gamma \parof{x \land S}) = %
    \weight{i}(x) \exp{\frac{\tau \gamma \log m}{\eps} \parof{A \parof{
    x \land S}}_i}.
  \end{align*}
  We have
  \begin{align*}
    \rip{\weight(x')}{\ones}  - \rip{\weight(x)}{\ones}
    &=                         %
      \rip{\varweight{\tau}}{\ones} - \rip{\varweight{0}}{\ones}
      =                         %
      \int_0^{\delta} \frac{d}{d \tau} \rip{ \varweight{\tau}}{\ones} \, d\tau
    \\
    &
      =                        %
      \gamma \frac{\log m}{\eps} %
      \int_0^{\delta} \pripover{A}{\varweight{\tau}}{x
      \land S} \, d\tau
      \tago{\leq}                        %
      \epsmore \gamma \frac{\log m}{\eps} %
      \int_0^{\delta} \pripover{A}{\weight}{x \land S} \, d\tau
      \\
      &
          \tago{\leq}
          \epsmore \frac{\log m}{\eps} \int_0^{\delta} \rip{\weight}{\ones} \, d\tau
      =                         %
      \delta \epsmore \frac{\log m}{\eps} \rip{\weight}{\ones}
  \end{align*}
  by \tagr monotonicity of $\varweight$ and
  \reflemma{pmwu-weight-step} and \tagr choice of $\gamma$.
\end{proof}

\begin{lemma}
  \labellemma{pmwu-packing}
  The output of the algorithm $\hat{x}$ satisfies
  $A \hat{x} \le (1+3\eps) \ones$.
\end{lemma}
\begin{proof}
  We prove a slightly stronger claim; namely, that at each time $t$,
  one has $A x \leq \parof{\epsmore t + 2 \eps \ones}
  \ones$.

  Consider \reflemma{pmwu-weight-growth}. So long as
  $A x \leq 2 \ones$, by interpolating (the upper bound on)
  $\rip{\weight}{\ones}$ as a continuous function of $t$, we have
  \begin{align}
    \frac{d}{d t} \rip{\weight}{\ones} %
    \leq                               %
    \epsmore \frac{\log m}{\eps} \rip{\weight}{\ones}.
  \end{align}
  Initially, when $t = 0$, we have
  \begin{math}
    A x %
    \leq %
    \eps \ones
  \end{math}
  by choice of $x$.

  Solving the differential inequality \reflastequation with initial
  value $m^2$, we have
  \begin{align*}
    \rip{\weight}{\ones} \leq m^2 \exp{\epsmore \frac{\log m}{\eps} t}
    =                           %
    \exp{\frac{\log m}{\eps} \parof{\epsmore t + 2 \eps}}.
  \end{align*}
  for all $t \in [0,1]$ as long as $A x \leq 2$. In particular, since
  \begin{math}
    \weight{i} =
    \exp{\frac{\log m}{\eps} (A x)_i} \leq \rip{\weight}{\ones}
  \end{math}
  for each $i$, we have
  \begin{align*}
    A x \leq \parof{\epsmore t + 2 \eps} \ones \leq \parof{1 + 3 \eps}
    \ones \leq 2 \ones.
  \end{align*}
  By induction on $t$, we have $A x \leq \parof{1 + 3 \eps} \ones$ for
  all $t \in [0,1]$.
\end{proof}

\subsection{Approximation ratio}
\labelsection{pmwu-apx}

We now analyze the approximation ratio of the output solution $\xout$.
The main observation, similar to \reflemma{pg-threshold} for
\algo{parallel-greedy}, is that $\lambda$ is an upper bound
on the gap $\opt - F(x)$.
\begin{lemma}
  At all times, $\lambda \geq \opt - F(x)$.
\end{lemma}

\begin{proof}
  The claim holds initially with $\lambda \geq \opt$ and
  $\mlf{x} \geq 0$. Whenever $x$ is increased and $\lambda$ is
  unchanged, $\mlf{x}$ increases due to monotonicity of $\mlf$, hence
  the claim continuous to hold.  Whenever $\lambda$ is about to be
  decreased in \refstep{pmwu-threshold}, we have $S$ empty with
  respect to the current value of $\lambda$.  Thus, letting $z$ be an
  optimal solution, we have
  \begin{align*}
    \opt - F(x)                 %
    &\tago{\leq}                        %
      \mlf{z}{x}                  %
      \tago{\leq}                        %
      \rip{\dmlf{x}}{z}           %
      \tago{\leq}                        %
      \epsless^3 \lambda \frac{\ripover{A}{\weight}{z}}{W}
      +                           %
      \eps \epsless \lambda
    \\
    &\tago{\leq}                        %
      \epsless^3 \lambda \frac{\rip{\weight}{\ones}}{W}
      + \eps \epsless \lambda
      \tago{\leq}                        %
      \epsless^2 \lambda + \eps \epsless \lambda
      \leq                        %
      \epsless \lambda
  \end{align*}
  by \tagr monotonicity of $\mlf$, \tagr nonnegative concavity, \tagr
  $S = \emptyset$, \tagr $A z \leq \ones$, and \tagr
  $\epsless \rip{\weight}{\ones} \leq W$. Thus, after replacing $\lambda$
  with $\epsless \lambda$, we still have $\lambda \geq \opt - F(x)$.
\end{proof}

\begin{lemma}
  \labellemma{pmwu-apx}
  The output $\xout$ of the algorithm satisfies
  \begin{math}
    \mlf{\xout} \geq \apxless \parof{1 - e^{-1}} \opt.
  \end{math}
\end{lemma}
\begin{proof}
  From the preceding lemma and line
  \refstep{pmwu-gradient-condition} of the algorithm,
  we have the following. Suppose $x$ changes to $x'$ with step size
  $\delta$. We have
  \begin{align*}
    \mlf{x'} - \mlf{x} \ge (1-\eps)^4 \delta \lambda
    \ge (1-\eps)^4 \delta \parof{\opt - \mlf{x}},
  \end{align*}
  and $t$ increases by $\delta$.  Therefore, $\mlf{x}$, as a function
  of $t$, increases at a rate such that
  \begin{math}
    \mlf{x} \geq \parof{1 - e^{-\apxless t}} \opt.
  \end{math}
  In particular, since the algorithm terminates with either $\lambda
  \leq e^{-1} \opt$ or $t \geq 1$, the output $\xout$ satisfies
  \begin{math}
    \mlf{\xout} \geq \apxless \parof{1 - e^{-1}} \opt.
  \end{math}
\end{proof}

\subsection{Iteration count}
\labelsection{pmwu-iterations}

In this section, we analyze the total number of iterations in
\algo{parallel-mwu}. \algo{parallel-mwu} consists of two nested loops:
an outer loop \refsubsteps{pmwu-loop}, where $W$ and $\lambda$ are
adjusted, and an inner loop \refsubsteps{pmwu-inc-loop}, which
increases $x$ uniformly along ``good'' coordinates w/r/t fixed values
of $W$ and $\lambda$.  We first analyze the number of iterations of
the outer loop.
\begin{lemma}
  In each iteration of the outer loop \refsubsteps{pmwu-loop} except
  for the last, either $\lambda$ decreases by a
  $\epsless$-multiplicative factor, or $\rip{\weight}{\ones}$
  increases by a $\frac{1}{1-\eps}$-multiplicative factor. This
  implies that, since $\lambda$ ranges from $\opt$ to $e^{-1} \opt$,
  and $\rip{\weight}{\ones}$ ranges from $m$ to $m^{\bigO{1/\eps}}$,
  there are at most $\bigO{\frac{\log m}{\eps^2}}$ iterations of the
  outer loop.
\end{lemma}
\begin{proof}
  The inner loop \refsubsteps{pmwu-inc-loop} terminates when either
  $S = \emptyset$, $\epsless \rip{\weight}{\ones} = W$, or $t = 1$. If
  $S = \emptyset$, then $\lambda$ decreases in line
  \refstep{pmwu-threshold}. If $\epsless \rip{\weight}{\ones} = W$,
  then since $W$ was the value of $\rip{\weight}{\ones}$ at the
  beginning of the iteration, we have that $\rip{\weight}{\ones}$
  increased by a $\prac{1}{1 - \eps}$-multiplicative factor. If
  $t = 1$, then this is the last iteration of
  \refsubsteps{pmwu-loop}.
\end{proof}

We now analyze the number of iterations of the inner loop
\refsubsteps{pmwu-inc-loop} for each fixed iteration of the outer loop
\refsubsteps{pmwu-loop}. Each iteration of the inner loop, except for
possibly the last, fixes $\delta$ based on either
\refstep{pmwu-gradient-condition} or
\refstep{pmwu-weight-condition}. We first bound the number of times
$\delta$ can be chosen based on \refstep{pmwu-weight-condition}. The
key idea to the analysis (due to \cite{young-01}) is that one can only
geometrically increase the coordinates in $S$ a small number of times
before violating the upper bounds on the coordinates implied by
$A x \leq \apxmore \ones$.

\begin{lemma}
  \labellemma{pmwu-weight-step-count}
  In each iteration of the outer loop \refsubsteps{pmwu-loop},
  $\delta$ is determined by \refstep{pmwu-weight-condition} at most
  $\bigO{\frac{\log n \log m}{\eps^2}}$ times.
\end{lemma}
\begin{proof}
  If $\delta$ is determined by \refstep{pmwu-weight-condition} more
  than $\bigO{\frac{\log n \log m}{\eps^2}}$ times, consider the
  coordinate $j$ that survives in $S$ throughout all of these many
  iterations. Such a coordinate exists because the set $S$ is
  decreasing throughout the inner loop.  The initial value of $x_j$
  sets $A_{ij} x_j \geq \eps / n$ for some $i \in [m]$, and by
  \reflemma{pmwu-packing}, $A_{ij} x_j \leq (A x)_i$ cannot exceed
  $1 + \bigO{\eps}$. Each iteration where $\delta$ is determined by
  \refstep{pmwu-weight-condition} increases $x_j$ (hence $A_{ij} x_j$)
  by a $\parof{1 + \bigOmega{\frac{\eps^2}{\log m}}}$-multiplicative
  factor.  Applying this multiplicative increase more than
  $\bigO{\frac{\log n \log m}{\eps^2}}$ times would violate the upper
  bound on $A_{ij} x_j$.
\end{proof}

We now analyze the number of times $\delta$ can be chosen based on
\refstep{pmwu-gradient-condition}. The following lemma is analogous to
\reflemma{pg-depth}, but the analysis is more subtle because of (a)
the general complexity added by the weights and (b) the fact that the
underlying potential function is not monotone.
\begin{lemma}
  \labellemma{pmwu-gradient-step} For a fixed iteration of the outer
  loop \refsubsteps{pmwu-loop}, $\delta$ is determined by
  \refstep{pmwu-gradient-condition} at most
  $\bigO{\frac{\log n}{\eps}}$ times.
\end{lemma}
\begin{proof}
  The overall proof is based on the potential function
  $\rip{\dmlf{x}}{x \land S}$, which is always in the range
  $\poly{\frac{\eps}{n}} \opt \leq \rip{\dmlf{x}}{x \land S} \leq
  \poly{n} \opt$ for nonempty $S$. An important distinction from the
  potential function of \reflemma{pg-depth} is that
  $\rip{\dmlf{x}}{x \land S}$ is not monotone: $\dmlf{x}$ and $S$ are
  both decreasing, but \emph{$x$ is increasing}. On one hand, the
  total growth by $x$ is bounded above by an
  $\bigO{n / \eps}$-multiplicative factor coordinatewise by our
  initial choice of $x_j$, as discussed in
  \reflemma{pmwu-weight-step-count}. On the other hand, we claim that
  whenever $\delta$ is determined by
  \refstep{pmwu-gradient-condition}, $\rip{\dmlf{x}}{x \land S}$
  decreases by a $\parof{1 - \bigOmega{\eps}}$-multiplicative
  factor. We prove the claim below, but suppose for the moment that
  this claim holds. Then we have a $\poly{n / \eps}$-multiplicative
  range for $\rip{\dmlf{x}}{x \land S}$ with non-empty $S$, and that
  the total increase of $\rip{\dmlf{x}}{x \land S}$ is bounded above
  (via the bound on the growth of $x$) by a
  $\bigO{\frac{n}{\eps}}$-multiplicative factor. It follows that
  $\delta$ is determined by \refstep{pmwu-gradient-condition} at most
  $\bigO{\frac{\log n}{\eps}}$ times until $\rip{\dmlf{x}}{x \land S}$
  falls below the lower bound $\poly{\eps/n}$, and $S$ is empty.

  Now we prove the claim. Let
  $x' = x + \delta \gamma \parof{x \land S}$, $w' = w(x')$, and $S'$
  denote the values of $x$, $w$ and $S$ after the update in step
  \refstep{pmwu-update}.  We want to show that
  $\rip{\dmlf{x'}}{x' \land S} \leq \parof{1 - c \eps}
  \rip{\dmlf{x}}{x \land S}$ for some constant $c > 0$.  We have
  \begin{align*}
    \delta \gamma \rip{\dmlf{x'}}{x' \land S'}
    &\tago{\leq}                        %
      \parof{1 + \frac{\eps^2}{\log m}} %
      \delta \gamma \rip{\dmlf{x'}}{x \land S'}
    \\
    &\tago{\leq}
      \parof{1 + \frac{\eps^2}{\log m}}
      \delta \gamma \rip{\dmlf{x'}}{x \land S}
    \\
    &\tago{\leq}
      \parof{1 + \frac{\eps^2}{\log m}}
      \mlf{x'}{x}
          \tago{=}                      %
          \parof{1 + \frac{\eps^2}{\log m}}
          \epsless^4 \delta \lambda    %
    \\
    &
      \tago{=}                   %
      \parof{1 + \frac{\eps^2}{\log m}}
      \epsless^4 \delta \gamma \frac{\pripover{A}{\weight}{x \land
      S}}{W}
    \\
    &
      \tago{\leq}
      \parof{1 + \frac{\eps^2}{\log m}} \epsless \delta \gamma
      \rip{\dmlf{x}}{x \land S}
    \\
    &\leq                %
      \parof{1 - \bigOmega{\eps}} \delta \gamma \rip{\dmlf{x}}{x \land S}
  \end{align*}
  by \tagr $x' \leq \parof{1 + \frac{\eps^2}{\log m}} x$ by choice of
  $\delta \gamma$, \tagr nonnegativity of
  $\dmlf$ and $S' \subseteq S$, \tagr monotonic concavity of $\mlf$,
  \tagr choice of $\delta$, \tagr choice of $\gamma$, and \tagr
  definition of $S$. Dividing both sides by $\delta \gamma$ gives the
  inequality we seek.
\end{proof}

\begin{lemma}
  Each step of \algo{parallel-mwu} has at most
  $\bigO{\frac{\log^2 m \log n}{\eps^4}}$ iterations.
\end{lemma}
\begin{proof}
  The preceding lemmas show that we have
  $\bigO{\frac{\log m}{\eps^2}}$ iterations of the outer loop and
  $\bigO{\frac{\log m \log n}{\eps^2}}$ iterations of the inner loop per
  iteration of the outer loop.
\end{proof}

\subsection{Number of oracle calls and additional implementation details}
\labelsection{pmwu-work}

In this section, we briefly account for the total work and number of
oracle calls of the algorithm.  The bottlenecks are step
\refstep{pmwu-greedy-step-size}, where we search for a value of
$\delta$ satisfying certain constraints, and \refstep{pmwu-update},
where one updates $\weight$ and $S$.

\paragraph{Estimating $\opt$:} The algorithm requires a value
$\lambda$ that is an upper bound on $\opt$. Preprocessing allows us to
choose $\lambda = \sum_j f(j)$ and we have $\opt \le \lambda \le n
\opt$.  It is also useful to have an estimate that is within a
constant factor. This can be done (a standard idea) by running the
algorithm in parallel with $O(\log n)$ geometrically increasing values
of $\lambda$ and picking the best solution from the parallel runs.

\paragraph{Step size:}
We first note that the greedy step size $\delta$ does not have to be
computed exactly, and that a $\epsmore$-multiplicative factor
approximation suffices. Indeed, suppose the algorithm is at step
\refstep{pmwu-greedy-step-size}. Let $\delta$ be the exact maximum
value satisfying the conditions \refsubsteps{pmwu-greedy-step-size},
and let $\apxdelta$ be a value such that
$\delta \leq \apxdelta \leq \min{\epsmore \delta, \frac{\eps^2}{\log
    m}}$. We want to show that $\apxdelta$ approximately satisfies
\refsubsteps{pmwu-greedy-step-size}. Indeed, we have
\begin{align*}
  \mlf{x + \apxdelta \gamma (x \land S)}{x} %
  \tago{\geq}                                    %
  \mlf{x + \delta \gamma (x \land S)}{x}  %
  \tago{\geq}                                    %
  \epsmore^4 \delta \lambda
  \geq                          %
  \frac{\epsless^4}{1 + \eps} \apxdelta \lambda.
\end{align*}
by \tagr monotonicity, \tagr choice of $\delta$, and \tagr choice of
$\apxdelta$. The \emph{inequality} \refstep{pmwu-gradient-condition}
is invoked only in \reflemma{pmwu-apx}. It is easy to see that the
slightly weaker inequality w/r/t $\apxdelta$ is enough to prove
\reflemma{pmwu-apx} with only a change in the hidden constants. The
other point where \refstep{pmwu-gradient-condition} is invoked in the
proof is when $\delta$ is determined by
\refstep{pmwu-gradient-condition}. Here we need only observe that
increasing $x$ further along $x \land S$ with a larger step size
$\apxdelta > \delta$ only decreases the coordinate values $\mlf{x}$
and thereby $\rip{\dmlf{x}}{S}$.

The second claim is that one needs only guess $\bigO{\log{n} / \eps}$
values of $\delta$. Indeed, we know that
$\gamma \delta \leq \eps^2 / \log m$. On the other hand, if
$\gamma \delta \leq \poly{\eps / n}$, then it is easy to show that
\refstep{pmwu-gradient-condition} is still satisfied. Thus one only
needs to check try $\apxdelta$ for $\bigO{\log{n} / \eps}$ powers of
$\epsmore$ between between $\poly{\eps/n}$ and $\eps^2 / \log m$.

\paragraph{Oracle calls to $\mlf$ and $\dmlf$:}
Now, to execute step \refstep{pmwu-greedy-step-size}, we must evaluate
$\mlf{x'}$ for $\bigO{\log{n} / \eps}$ different possible choices of
$\delta$. To execute step \refstep{pmwu-update}, we need $\apxO{N}$
work to update $\weight$, and then obtain a partial derivative
$\dmlf[j]{x}$ for each coordinate $j \in [n]$. Since the depth of
either step is $\bigO{\frac{\log^2 m \log n}{\eps^2}}$, we see that
the algorithm requires $\apxO{N / \eps^2}$ total work (where $N$ is
the total number of nonzeroes in $A$), $\bigO{\frac{\log^2 m \log^2
    n}{\eps^2}}$ calls to evaluate $\mlf{x}$, and $\bigO{n
  \frac{\log^2 m \log^2 n}{\eps^2}}$ to individual coordinates of
$\dmlf{x}$. It is not hard to see that if we have a $(1+\eps)$
multiplicative approximation for these quantities then the whole
analysis still goes through.

\paragraph{Oracle calls to $f$:}
If we assume only oracle access to the underlying submodular function
$f$, then we need to estimate $\mlf{x}$ and $\dmlf[j]{x}$ based on
random sampling. Each sample constitutes a query to $f$ and the
samples are done in parallel followed by aggregation. We note that the
starting value of $x$ in the algorithm satisfies the property that
$\mlf{x} \ge \frac{\eps}{n} \opt$.  Moreover, we can assume at any
point that $\dmlf[j]{x} \ge \frac{\eps}{n}\opt$; for if it is smaller
then even taking all such coordinates to $1$ would contribute at most
$\eps \opt$. Following the discussion after \reflemma{dmlf-sample} we
see that one can obtain an estimate for $\dmlf[j]{x}$ that is, with
high probability, a $(1+\eps)$ multiplicative approximation using
$\apxO{n/\eps^3}$ samples.  Similarly with $\apxO{n/\eps^3}$ samples
one can get a $(1+\eps)$-approximation for $\mlf{x}$ with high
probability. From the preceding analysis the algorithm makes
$\apxO{n/\eps^2}$ calls to individual coordinates of $\dmlf$ and
$\apxO{1/\eps^2}$ calles to $\mlf$. Thus the total number of oracle
calls to $f$ is $\apxO{n^2/\eps^5}$. Since the correctness
probability of each estimate is inversely polynomial in $n$, we can
take a union bound over the $\apxO{n/\eps^2}$ estimates that the algorithm
requires.


\section{Rounding the fractional solution in parallel}
\labelsection{rounding} In this section we briefly discuss, at a
high-level, a few settings in which one can round the fractional
solution to the multilinear relaxation in parallel. We assume some
familiarity with prior work on rounding the multilinear relaxation in
the sequential setting, and we also restrict our attention to monotone
functions. Formal details are outside the scope of this paper.  Let
$x$ be a feasible fractional solution to constraints of the form $Ax
\le b, x \in [0,1]^n$ where $A,b$ have non-negative coefficients.

\paragraph{Cardinality constraint:} First consider the case where we have
a simple cardinality constraint of the form $\sum_i x_i \le k$.
If $k = O(\log n/\eps^2)$ we can do sequential greedy, hence we
assume $k = \Omega(\log n/\eps^2)$. In this case we can
pick each element $i$ independently with probability $(1-\eps)x_i$.
Then with high probability we will not violate the constraint due
to Chernoff bounds. Moreover independent rounding also preserves $F$.
In other words if $R$ is a random set obtained via the rounding
we have $\evof{f(R)} \ge (1-\eps) F(x)$; further $f(R)$ is
concentrated around $(1-\eps)F(x)$ \cite{v-10}. This allows us
to obtain a $(1-1/e-\eps)$ approximation via rounding.

\paragraph{Packing constraints:} For general packing constraints one
can round via a contention resolution scheme (CRS) to obtain an
approximation ratio of $\Omega(1/\Delta)$ where $\Delta$ is the
maximum number of non-zeroes in any column; we refer the reader to
\cite{BansalKNS,cvz-14}. The scheme is very simple and consists of
independently rounding each $i$ with probability $c x_i/\Delta$ for
some constant $c < 1$ and then doing an alteration step on the
resulting random set $R$ to make it feasible. The alteration step is
composed of independent alteration steps for each constraint.  A
cursory glance at the details of the alteration step would suffice to
convince oneself that it can be easily parallelized using sorting. We
note that $\Delta = O(1)$ for several simple constraints of
interest. As examples, if $A$ corresponds to the constraints of a
partition matroid we have $\Delta = 1$, for matchings and
$b$-matchings $\Delta = 2$. We also note that the approximation ratio
improves as the width of $A$ increases \cite{BansalKNS,cvz-14}.

\paragraph{Partition matroid constraint:} The CRS scheme for general
packing constraints gives a constant factor for partition matroid
constraints since $\Delta = 1$. However, it is known that any
fractional point $x$ in a matroid polytope can be rounded to an
integral solution without any loss. The two known techniques to
achieve this are pipage rounding \cite{ccpv} and swap rounding
\cite{cvz-swap-rounding}. It is not clear how to parallelize these
schemes for general matroids but partition matroid constraints are
simple. One can implement, with some tedious technical work,
swap-rounding in poly-logarithmic depth for partition matroids.  Here
we give another approach which is simple to describe. Let the
partition matroid over $\groundset$ be defined by the partition
$\groundset_1,\groundset_2,\ldots,\groundset_h$, with $k_j$ indicating
the number of elements that can be chosen from $\groundset_j$.  If
$k_j = 1$ for $1 \le j \le h$ then we have a simple partition
matroid. For simple partition matroids the rounding is easy. We have a
constraint for each $\groundset_j$ of the form $\sum_{i \in
  \groundset_j} x_j \le k_j = 1$. In this case, independently and in
parallel for each $\groundset_j$, we pick at most one element where
the probability of picking $i \in \groundset_j$ is precisely equal to
$x_i$. The random set $R$ output by this algorithm satisfies the
property that $\evof{f(R)} = F(x)$, and clearly satisfies the
constraints. One can reduce the problem of maximizing a submodular
function over a partition matroid to maximizing over a simple
partition matroid via the following lifting trick that is well-known.
It is easier to first explain the reduction for the case of a
cardinality constraint. Suppose $f: \groundsets \to \reals$ and we
wish to solve the problem $\max_{|S| \le k} f(S)$. This is a special
case of a partition matroid.  We create a new ground set $\groundset'
= \groundset \times\setof{1,2,\ldots,k}$ which corresponds to creating
$k$ copies of each element $e \in \groundset$.  Let $\groundset'_j =
\{(e,j) \mid e \in \groundset\}$.  Consider a derived submodular
function $g: 2^{\groundset'} \to \reals$ defined as follows. For $A
\subseteq \groundset'$ we define its projection to $\groundset$,
denoted by $A_{\groundset}$, as the set $\{e \in \groundset \mid
\exists i, (e,i) \in A\}$; that is, we collapse all the copies of an
element $e$ to $e$.  We then define $g$ by settings $g(A) =
f(A_{\groundset})$ for each $A \subseteq \groundset'$. It is
relatively easy to verify that $g$ is monotone submodular if $f$ is
monotone submodular. Maximizing $f$ subject to a cardinality
constraint is equivalent to maximizing $g$ with a simple partition
matroid constraint over $\groundset'$ where the partition is
$\groundset'_1,\ldots,\groundset'_k$. A value oracle for $g$ is
straight forward given a value oracle for $f$. Thus we have reduced
the cardinality constrained problem into a simple partition matroid
constraint. One can apply this lifting trick to each partition of a
general partition matroid to reduce the problem to a simple partition
matroid constraint. Note that one only needs to lift parts with
capacity at most $\bigO{\log{n} / \eps^2}$, since otherwise randomized
rounding suffices (as in the cardinality case above).

We believe that one can, with some technical work, also round a
fractional solution for a laminar matroid constraint in
poly-logarithmic depth by suitably adapting swap rounding.


\bibliographystyle{plainnat}
\bibliography{submodular-mwu}
\appendix

\section{Analysis of \algo{randomized-parallel-greedy}}

\labelappendix{rpg}

In this section, we analyze the \algo{randomized-parallel-greedy}
algorithm given in \reffigure{randomized-parallel-greedy} and obtain
the bounds of \reftheorem{rpg}. As mentioned when introducing the
algorithm earlier, the basic idea is that
\algo{randomized-parallel-greedy} preserves the most important
invariants of \algo{parallel-greedy} in expectation.

The variable $t$ is used to track the expected cardinality of
$\sizeof{Q}$. One can simplify the algorithm by replacing the variable
$t$ with $\sizeof{Q}$ wherever it appears, and the analysis would
generally still hold. (In this case, a variable similar to $t$ could
be introduced in the analysis.) The inclusion of $t$ makes the
analysis more straightforward, as we can use $t$ to track progress
throughout the algorithm.

Intuitively, the key points to the analysis of \algo{parallel-greedy}
now hold in expectation.  The only source of randomization is in
\refstep{rpg-sample}, where we sample a set $R \sim \delta Q$ to add
to $S$. Let us assume that $k \geq C \frac{\log n}{\eps^2}$ for some
constant $C > 1$, since otherwise one can simply run the sequential
greedy algorithm and obtain the desired depth. Then the cardinality of
the computed solution $Q$ is randomized, but is concentrated at
$t \leq \epsless k$ because each independent coin toss adjusts the
cardinality by at most 1, and $t$ is at least
$\epsless C \frac{\log n}{\eps^2}$.

\begin{lemma}
  \labellemma{rpg-cardinality} %
  With high probability, we have $\sizeof{S} \leq \epsmore t \leq k$
  throughout the algorithm.
\end{lemma}
\begin{proof}
  We assume that $k \geq c \frac{\log n}{\eps^2}$ for a large constant
  $c > 0$, since otherwise one can simply run the sequential greedy
  algorithm instead. In this case we also have
  $t \geq c \frac{\log n}{\eps^2}$ throughout the algorithm. As
  discussed above, $t$ tracks the expected cardinality of $Q$, summing
  the expected increase $\sizeof{R} = \delta \sizeof{S}$ over each
  iteration of \refstep{rpg-sample}. On the other hand, each random
  coin toss from sampling $R \sim \delta S$ affects the cardinality of
  $Q$ by at most $1$. Since the expected cardinality of $Q$ is
  $t \geq c \frac{\log n}{\eps^2}$, this is at most a
  $\frac{\eps^2 }{c \log n}$-fraction of the expected total. It
  follows from (online extensions of) the multiplicative Chernoff
  inequality that $\sizeof{Q} \leq \epsmore t$ all throughout the
  algorithm.

  To make this argument formal, for $\ell \in \naturalnumbers$ and
  $j \in \groundset$, let $X_{\ell,j} \in \setof{0,1}$ indicate
  whether or not $j$ is sampled by $R$ in the $\ell$th iteration of
  \refstep{rpg-sample}. Each $X_{\ell,j}$ depends on $X_{\ell',j'}$
  for previous iterations $\ell' < \ell$ and $j \in [n]$, but once
  these outcomes are fixed, each $X_{\ell,j}$ is independent of the
  other indicator variables in the $\ell$th iteration.  Let
  \begin{math}
    Y_{\ell,j} = \evof{X_{\ell,j} \given
      X_{\ell',j'} \text{ for } \ell' < \ell, j \in j'}
  \end{math}
  by the expected value of $X_{\ell,j}$ \emph{going in to the $\ell$th
    round}.  If we let $S_{\ell}$ and $\delta_{\ell}$ denote the
  values of the set $S$ and step size $\delta$ during the $\ell$th
  iteration of \refstep{rpg-sample}, then we have
  $Y_{\ell,j} = \delta$ if $j \in S_{\ell}$, and 0 otherwise. In
  particular, we have
  \begin{math}
    \sum_{\ell\in \naturalnumbers} \sum_{j \in [n]} Y_{\ell,j} =
    \sum_{\ell \in \naturalnumbers} \delta_{\ell} \sizeof{S_{\ell}} =
    t \leq \parof{1 - 2 \eps} k
  \end{math}
  deterministically. By online extensions of the Chernoff inequality
  (\reflemma{online-chernoff}), we have
  \begin{align*}
    \probof{\sum_{\ell,e} X_{\ell,e} \geq k} %
    \leq                                     %
    \probof{\epsmore \sum_{\ell,e} Y_{\ell,e} + \eps k}
    \leq                        %
    \probof{\epsmore \sum_{\ell,e} Y_{\ell,e} + \bigO{\frac{\log n}{\eps}}}
    \leq                        %
    \poly{1/n},
  \end{align*}
  as desired.
\end{proof}

\begin{lemma}
  \labellemma{rpg-apx} At each time $t$, we have
  \begin{math}
    \evof{f(Q)} \geq \apxless \parof{1 - e^{-t/k}} \opt.
  \end{math}
  In particular, since the algorithm exits with either
  $t = (1 - 2 \eps) k$ or $\opt - f(\Qout) \leq e^{-1} \opt$, the final
  set $\Qout$ satisfies
  \begin{math}
    \evof{f(\Qout)} \geq \apxless \parof{1 - e^{-1}} \opt.
  \end{math}
\end{lemma}

\begin{proof}
  Given $Q$, We have
  \begin{align*}
    \evof{f_Q(R) \given Q}           %
    \tago{=}                           %
    \mlf{Q + \delta S}{Q}
    \tago{\geq}                        %
    \epsless^2 \frac{\delta \sizeof{S}}{k} \parof{\opt - f(Q)}
  \end{align*}
  by \tagr definition of the multilinear extension and \tagr
  \reflemma{pg-threshold}. Taking expectations over $Q$
  \begin{align*}
    \evof{\frac{d f(Q)}{d t}} %
    \geq                               %
    \evof{\frac{\epsless^2}{k} \parof{\opt - f(Q)}}
    =                           %
    \frac{\epsless^2}{k} \parof{\opt - \evof{f(Q)}}.
  \end{align*}
  By Fubini's theorem, one can interchange the expectation and the
  derivative, which gives
  \begin{align*}
    \frac{d}{dt} \evof{f(Q)}
    =
    \evof{\frac{d}{dt} f(Q)}    %
    \geq                        %
    \frac{\epsless^2}{k} \parof{\opt - \evof{f(Q)}}.
  \end{align*}
  Solving the differential inequality, we have
  \begin{align*}
    \evof{f(Q)}
    \geq                        %
    \parof{1 - e^{- \epsless^2 t / k}} \opt
    \geq
    \apxless \parof{1 - e^{-t / k}} \opt
    ,
  \end{align*}
  as desired.
\end{proof}

\begin{lemma}
  \labellemma{rpg-depth}
  If
  \begin{math}
    \mlf{Q + \delta S}{Q} = \epsless^2 \lambda \frac{\delta \sizeof{S}}{k},
  \end{math}
  then $\sizeof{S}$ decreases by a $\epsless$-multiplicative factor in
  expectation. This implies that, for fixed $\lambda$, the loop at
  \refsubsteps{rpg-inc-loop} iterates at most
  $\bigO{\frac{\log n}{\eps}}$ times in expectation (via
  \reflemma{concentration-decay}). In total, the loop at
  \refsubsteps{rpg-inc-loop} iterates at most
  $\bigO{\frac{\log n}{\eps^2}}$ times in expectation.
\end{lemma}
\begin{proof}
  We have
  \begin{align*}
    \epsless^2 \lambda \frac{\delta \sizeof{S}}{k} %
    &\tago{=}                                              %
      F_Q(Q + \delta S)                                    %
      \tago{\geq}                       %
      \delta \rip{\dmlf{Q + \delta S}}{S} %
    \\
    &\tago{=}
      \delta \sum_{j \in S} \evof{f_{Q \cup R}(j)} %
          \tago{=}
          \delta \evof{\sum_{j \in S} f_{Q \cup R}(j)}
    \\
    &\tago{\geq}               %
      \delta \evof{\sum_{j \in S'} f_{Q \cup R}(j)}
      \tago{\geq}
      \delta \evof{\sizeof{S'} \frac{\epsless \lambda}{k}} %
      =                                                    %
      \epsless \lambda \frac{\delta}{k} \evof{\sizeof{S'}}
      \labelthisequation{rpg-depth-derivation}
  \end{align*}
  by \tagr choice of $\delta$, \tagr monotonic concavity, \tagr
  definition of $\dmlf$, \tagr linearity of expectation, \tagr
  monotonicity, and \tagr definition of $S'$.  Dividing both sides by
  $\eps \lambda \delta / k$, we have
  \begin{math}
    \evof{\sizeof{S'}} \leq \epsless \sizeof{S}
  \end{math}
  in expectation.
\end{proof}

We now prove \reftheorem{rpg}.
\begin{proof}
  Let $\mu = \evof{f(Q)} \geq \apxless \parof{1 - e^{-1}} \opt$, per
  \reflemma{rpg-apx}. We first show that $f(Q) \geq \mu$ with high
  probability. Let $p = \probof{f(S) \geq \epsless \mu}$. We first
  observe that
  \begin{align*}
    \probof{f(Q) \geq \opt}     %
    \tago{\leq}                        %
    \probof{\sizeof{Q} \geq k}  %
    \tago{\leq}                        %
    \frac{1}{\poly{n}}
  \end{align*}
  because \tagr $f(Q) \leq \opt$ whenever $\sizeof{Q} \leq k$ and
  \tagr by the concentration bound \reflemma{rpg-cardinality}. We have
  \begin{align*}
    \mu = \evof{f(Q)}                 %
    &\tago{=}                           %
      p \evof{f(Q) \given f(Q) \leq \epsless \mu} %
    \\
    &\quad
      +                                           %
      \parof{ 1 - p - \frac{1}{\poly{n}}} \evof{f(Q) \given \epsless
      \mu \leq f(Q) \leq \mu} %
    \\
    &\quad+                           %
      \frac{1}{\poly{n}} \evof{f(Q) \given f(Q) \geq \mu} %
    \\
    &\tago{\leq}
      p \epsless \mu +
      \parof{1 - p - \frac{1}{\poly{n}}} \mu         %
      +
      \frac{1}{\poly{n}} n \mu    %
    \\
    &\leq                      %
      - \eps p \mu + \mu + \frac{\mu n}{\poly{n}}.
  \end{align*}
  by \tagr conditional expectations and \tagr bounding $f(Q)$ by
  $\epsless \mu$, $\mu$, and $n \mu$ respectively.  Dividing
  both sides by $\mu$ and rearranging, we have
  \begin{align*}
    p \leq \frac{n}{\eps \poly{n}} = \frac{1}{\poly{n}}
  \end{align*}
  for a slightly smaller polynomial $\poly{n}$, as desired.

  It remains to account for the depth and oracle calls, and work. The
  depth was proven \reflemma{rpg-depth}. The expected number of oracle
  calls is bounded by multiplying the expected depth by
  $\apxO{n/\eps^2}$, which is the number of random samples needed to
  estimate $\mlf{x}$.
\end{proof}


\section{Knapsack constraints}

\labelappendix{knapsack}

\begin{figure}[t]
  \centering
  \begin{framed}
    \ttfamily\raggedright
    \underline{parallel-greedy-knapsack($f$,$\groundset$,$a$)}
    \begin{steps}
    \item
      \begin{math}
        x \gets \setof{j \where a_j \leq \frac{\eps}{n}},
      \end{math}
      $\lambda \gets \opt$ %
      \commentcode{or any upper bound for $\opt$} %
    \item while $\rip{a}{x} \leq 1$ and $\lambda \geq e^{-1} \opt$
      \begin{steps}
      \item let
        \begin{math}
          S = \setof{j \in \groundset \where \dmlf[j]{x} \geq \epsless
            \lambda a_j}
        \end{math}
      \item while $S$ is not empty and $\rip{a}{x} \leq 1$
        \begin{steps}
        \item \labelstep{pgk-heavy} if
          \begin{math}
            a_j \geq \bigOmega{\frac{\eps^2}{\log n}}
          \end{math}
          for some $j \in S$
          \begin{steps}
          \item if $\rip{a}{x} + a_j \leq 1$
            \begin{steps}
            \item \labelstep{pgk-heavy-inc}
              $x_j \gets 1$
            \end{steps}
          \item \labelstep{pgk-heavy-exit} %
            else return $x$
          \end{steps}
        \item \labelstep{pgk-light} else
          \begin{steps}
          \item choose $\delta$ maximal s.t.\
            \begin{steps}
            \item
              \begin{math}
                \mlf{x + \delta S}{x} %
                \geq %
                \epsless^2 \lambda \delta \rip{a}{S}
              \end{math}
            \item
              \begin{math}
                \rip{a}{x + \delta S} \leq 1
              \end{math}
            \end{steps}
          \item
            \labelstep{pgk-light-inc}
            $x \gets x + \delta S$
          \end{steps}
        \end{steps}
      \item $\lambda \gets \epsless \lambda$
      \end{steps}
    \item return $x$
    \end{steps}
    \caption{A parallel implementation of the
      \algo{continuous-greedy} algorithm specialized to the
      knapsack polytope. \labelfigure{parallel-greedy-knapsack}}
  \end{framed}
\end{figure}

In this section, we consider the parallel continuous greedy algorithm
a single knapsack packing constraint, an intermediate setting in
between the cardinality constraint and general packing
constraints. Formally, we consider the following problem:
\begin{align*}
  \text{maximize } F(x) \text{ over } x \geq \ones \text{ s.t.\ }
  \rip{a}{x} \leq 1,
\end{align*}
where $a: \groundset \to [0,1]$ is a positive cost vector. Here we
have normalized the costs so that the size of the knapsack is 1. We
let $\infnorm{a} = \max_j a_j$ be the maximum cost of any item. We
first present algorithms that obtain approximation factor that depend
on $a$; the dependency can then be removed by partial enumeration
(without increasing the depth, but increasing the total amount of
work).

As with the cardinality constraint, we first consider a model where we
have oracle access to the multilinear extension $\mlf$ and its
derivatives $\dmlf$. We present an algorithm that is called
\algo{parallel-greedy-knapsack} and given in
\reffigure{parallel-greedy-knapsack}. It is very similar to
\algo{parallel-greedy}, and can be interpreted as a parallel extension
of the \algo{continuous-greedy} algorithm of \citet{ccpv} specialized
to the knapsack polytope. The primary differences from
\algo{parallel-greedy} are as follows. First, we simply take any
coordinate with cost at most $\eps / n$. This only uses an
$\eps$-fraction of the budget, and the fact that all remaining
coordinates have cost at least $\eps / n$ will be useful in the
analysis. The second difference is probably the most significant
difference, and is as follows. When gathering the set of ``good''
coordinates $S$, rather than comparing the partial derivative
$\dmlf[j]{x}$ of each coordinate to a fixed threshold, we compare the
``bang-for-buck ratio'' $\dmlf[j]{x} / a_j$ of the partial derivative
to the cost to the threshold. Third, when adding coordinates to our
solution, we take special exception for items whose costs are at least
a $\bigOmega{\eps^2 / \log n}$-fraction of the budget.
When a good coordinate $j$ has such a high cost or partial
derivative, we directly set $x_j = 1$ rather than take a fractional
amount. Maintaining the invariant $x_j \in \setof{0,1}$ for all
coordinates $j$ with $a_j \geq \eps^2 / \log n$ is convenient for
applying the Chernoff inequality should one want to round $x$ to a
discrete solution later.

The final bounds and proof are similar to that of the cardinality
constraints in \refsection{cardinality}, and many of the differences
are analogous to the differences between cardinality and knapsack
constraints in the well-known sequential setting. Consequently we
restrict ourselves to brief sketches of proofs, highlighting the main
differences from the proofs of \refsection{cardinality}.

\begin{lemma}
  \labellemma{pgk-threshold} At any point, we have
  $\lambda \geq \opt - \mlf{x}$. This implies the following.
  \begin{enumerate}
  \item In either step \refstep{pgk-heavy-inc} or
    \refstep{pgk-light-inc}, we have
    \begin{align*}
      \frac{d \mlf{x}}{d \rip{a}{x}} \geq \epsless^2 \parof{\opt -
      \mlf{x}},
      \labelthisequation{pgk-apx-deq}
    \end{align*}
    hence
    \begin{align*}
      \mlf{x} \geq \apxless \parof{1 - e^{-\rip{a}{x}}} \opt
    \end{align*}
    at any point.
  \item If $\lambda \leq e^{-1} \opt$, then
    \begin{math}
      \mlf{x} \geq \parof{1 - e^{-1}} \opt.
    \end{math}
  \end{enumerate}
  Thus, if the algorithm terminates with $\rip{a}{x} = 1$ or
  $\lambda \leq e^{-1} \opt$, then we have
  \begin{math}
    \mlf{x} \geq \apxless \parof{1 - e^{-1}} \opt.
  \end{math}
  Otherwise, the algorithm terminates in step
  \refstep{pgk-heavy-exit}, in which case there is an item
  $j \in \groundset$ such that
  \begin{math}
    \mlf{x} \geq \apxless \parof{1 - e^{-\parof{1 - a_j}}} \opt
  \end{math}
  and
  \begin{math}
    \mlf{x + j} \geq \apxless \parof{1 - e^{-1}} \opt.
  \end{math}
\end{lemma}

\begin{proof}[Proof sketch]
  The proof is the same as \reflemma{pg-threshold}, with the only
  change being that we now have
  \begin{align*}
    \rip{\dmlf{x}}{z}           %
    \tago{\leq}                 %
    \epsless \lambda \rip{z}{a} %
    \tago{\leq}                 %
    \epsless \lambda            %
  \end{align*}
  by \tagr emptiness of $S$ and \tagr $\rip{z}{a} \leq k$. We should
  note that, at the beginning of each iteration of step
  \refstep{pgk-heavy-inc}, we have
  \begin{math}
    \mlf{x + e_j} = \mlf{x} + \dmlf[j]{x}
  \end{math}
  by multilinearity of $\mlf$, which gives the differential inequality
  \refequation{pgk-apx-deq} when increasing $x$ in step
  \refstep{pgk-heavy-inc}.
\end{proof}

\begin{lemma}
  \labellemma{pgk-heavy-depth} \algo{parallel-greedy-knapsack} enters
  the \algo{if} clause \refsubsteps{pgk-heavy} at most
  $\bigO{\frac{\log n}{\eps^2}}$ times over the course of the
  algorithm.
\end{lemma}
\begin{proof}
  Each time we enter the if clause (except possibly the last), we
  increase $t$ by at least $\bigOmega{\frac{\eps^2}{\log n}}$. But we
  always have $t \leq 1$.
\end{proof}

\begin{lemma}
  \labellemma{pgk-light-depth}
  If
  \begin{math}
    \mlf{x + \delta S}{x} %
    = %
    \epsless^2 \lambda \delta \rip{a}{S},
  \end{math}
  then the step \refstep{pgk-light-inc} decreases $\rip{a}{S}$ by at
  least a $\epsless$-multiplicative factor. Since $\rip{a}{S}$ ranges
  from at most $n$ to at least $\eps / n$ (unless $S$ is empty), this
  implies that, for fixed $\lambda$, the steps \refsubsteps{pgk-light}
  iterate at most $\bigO{\log n / \eps}$ times, and at most
  $\bigO{\frac{\log n}{\eps^2}}$ times total.
\end{lemma}
\begin{proof}[Proof sketch]
  The proof is the same as in \reflemma{pg-depth}, except the
  modified definition of $S$ changes the endpoints of equation
  \refequation{pg-depth-derivation} to
  \begin{math}
    \epsless^2 \lambda \delta \rip{a}{S'}
  \end{math}
  and
  \begin{math}
    \epsless \lambda \delta \rip{a}{S''},
  \end{math}
  respectively.
\end{proof}
We conclude with a theorem summarizing the analysis. We note that the
upper bound on the depth comes from the combination of
\reflemma{pgk-heavy-depth} and \reflemma{pgk-light-depth}.  The oracle
complexity and total work follow from essentially the same analysis as
the cardinality setting.
\begin{theorem}
  \algo{parallel-greedy-knapsack} computes a vector $x$ with the
  following properties.
  \begin{enumerate}[label={\alph*}]
  \item $\rip{a}{x} \leq 1$
  \item
    \begin{math}
      \mlf{x} \geq \apxless \parof{1 - e^{-1 - \infnorm{a}}} \opt,
    \end{math}
    where $\infnorm{a} = \max_j a_j$ is the maximum cost of any item.
  \item If $\mlf{x} < \apxless \parof{1 - e^{-1}}\opt$, then there is
    an item $j \in \groundset$ such that
    \begin{math}
      f(x + j) \geq \parof{1 - \bigO{\eps}} \parof{1 - e^{-1}} \opt.
    \end{math}
  \item If $a_j \geq \frac{c \eps^2}{\log n}$ (for any desired
    constant $c > 0$), then $x_j \in \setof{0,1}$.
  \end{enumerate}
  \algo{parallel-greedy-knapsack} has depth
  $\bigO{\frac{\log n}{\eps^2}}$, uses
  $\bigO{\frac{n \log n}{\eps^2}}$ oracle calls to $\mlf{x}$ and
  coordinates of $\dmlf{x}$, and does total work
  $\apxO{\frac{n}{\eps^2}}$.
\end{theorem}

\subsection{Oracle complexity w/r/t $f$}

\begin{figure}[t]
  \centering
  \begin{framed}
    \ttfamily\raggedright
    \underline{randomized-parallel-greedy-knapsack($f$,$a$,$\eps$)}
    \begin{steps}
    \item $Q \gets \emptyset$, $t \gets 0$, $\lambda \gets \opt$
      \commentcode{or any upper bound for $\opt$}
    \item while $t \leq 1$ and $\lambda \geq e^{-1} \opt$
      \begin{steps}
      \item let $S = \setof{j \in \groundset \where f_Q(j) \geq
          \epsless \lambda a_j}$
      \item while $S$ is not empty and $t \leq 1 - \eps$
        \begin{steps}
        \item if
          \begin{math}
            a_j \geq \bigOmega{\frac{\eps^2}{\log n}}
          \end{math}
          or
          \begin{math}
            F_Q(j) \geq \bigOmega{\frac{\eps^2 \opt}{\log n}}
          \end{math} for some $j \in S$
          \begin{steps}
          \item if $t + a_j \leq 1 - \eps$
            \begin{steps}
            \item $Q \gets Q + j$, $t \gets t + a_j$
            \end{steps}
          \item else return $Q$
          \end{steps}
        \item else
          \begin{steps}
          \item
            choose $\delta > 0$ maximal s.t.\
            \begin{steps}
            \item
              \begin{math}
                \mlf{Q + \delta S}{Q} \geq \epsless^2 \delta
                \rip{a}{S}
              \end{math}
            \item $t + \delta \rip{a}{S} \leq 1 - \eps$
            \end{steps}
          \item \labelstep{rpgk-sample} sample $R \sim \delta S$
          \item
            $Q \gets Q \cup R$, $t \gets t + \delta \rip{a}{S}$
          \end{steps}
        \item $\lambda \gets \epsless \lambda$
        \end{steps}
      \end{steps}
    \item return $Q$
    \end{steps}
    \caption{A randomized, combinatorial variant of the previous
      \algo{parallel-greedy-knapsack} algorithm for knapsack
      constraints.\labelfigure{randomized-parallel-greedy-knapsack}}
  \end{framed}
\end{figure}

\begin{lemma}
  Let $\Qout$ be the final set $Q$ output by
  \algo{randomized-parallel-greedy-knapsack}. With high probability,
  \begin{math}
    \epsless t \leq \rip{a}{x} \leq \epsmore t \leq 1.
  \end{math}
\end{lemma}
\begin{proof}[Proof sketch]
  We use $t$ to track the expected size of $\rip{a}{Q}$. Every item
  with cost at least $c \eps^2 / \log n$, for some small constant
  $c > 0$, is treated deterministically. Conversely, every randomized
  decision contributes at most $c \eps^2 / \log n$ to $\rip{a}{Q}$. By
  online extensions of the Chernoff inequality
  (\reflemma{online-chernoff}, applied similarly to to the proof of
  \reflemma{rpg-cardinality}), the total cost is at most
  $\epsmore t \leq 1$ with probability $\geq 1 - 1 / \poly{n}$.
\end{proof}

\begin{lemma}
  \labellemma{rpgk-apx}
  At each time $t$,
  \begin{math}
    \evof{f(Q)} \geq \epsless^2 \parof{1 - \bigO{\eps}} \parof{1 -
      e^{-t / k}}.
  \end{math}
  In particular, since the algorithm exits with either $t = (1 - 2
  \eps) k$ or $\opt - f(S) \leq e^{-1} \opt$, the final set $\Qout$
  satisfies
  \begin{align*}
    \evof{f(\Qout)} \geq \apxless \parof{1 - e^{-1}} \opt.
  \end{align*}
\end{lemma}

\begin{proof}[Proof sketch]
  The proof essentially the same as \reflemma{rpg-apx}, except
  \reflemma{pg-threshold} is replaced by \reflemma{pgk-threshold}.
\end{proof}

\begin{lemma}
  \labellemma{rpgk-depth} If
  \begin{math}
    \dmlf{Q + \delta S}{Q} = \epsless^2 \delta \rip{a}{S}
  \end{math}
\end{lemma}

\begin{proof}[Proof sketch]
  The same proof as \reflemma{rpg-depth} goes through, except the
  modified definition of $S$ changes the first term of the derivation
  \refequation{rpg-depth-derivation} is replaced with
  \begin{math}
    \epsless^2 \lambda \delta \rip{a}{S}
  \end{math}
  and the last two terms of \refequation{rpg-depth-derivation} are
  replaced by
  \begin{align*}
    \cdots \geq \delta \evof{\epsless \lambda \rip{a}{S'}} \geq %
    \epsless \lambda \delta \evof{\rip{a}{S'}}.
  \end{align*}
  We note that the objective value is tightly concentrated because any
  item with margin $\geq \bigOmega{\frac{\eps^2 \opt}{\log n}}$ is
  decided deterministically.
\end{proof}

\begin{theorem}
  \algo{randomized-parallel-greedy-knapsack} returns a randomized set $Q$ such
  that
  \begin{enumerate}[label=(\alph*)]
  \item $\rip{a}{Q} \leq 1$ with high probability.
  \item
    \begin{math}
      \evof{f(Q)} \geq \apxless \parof{1 - e^{-(1 - \infnorm{a})}} \opt
    \end{math}
    and
    \begin{math}
      f(Q) \geq \epsless \evof{f(Q)}
    \end{math}
    with high probability, where $\infnorm{a} = \max_j a_j$ is the
    maximum cost of any item.
  \item
    \begin{math}
      \evof{\max_{j \in [n]} f(Q + j)} \geq \apxless \parof{1 -
        e^{-1}} \opt,
    \end{math}
    and
    \begin{math}
      \max_{j \in [n]} f(Q + j) \geq \evof{\max_{j \in [n]} f(Q + j)}
    \end{math}
    with high probability.
  \end{enumerate}
  \algo{randomized-parallel-greedy-knapsack} has depth
  $\bigO{\frac{\log n}{\eps^2}}$, and uses
  $\apxO{\frac{n}{\eps^2}}$ oracle calls to $f$.
\end{theorem}

\subsection{Partial enumeration}

Above we derived low-depth algorithms for knapsack constraints with an
approximation factor that degrades with the maximum cost of an
item. This suffices for many real applications, where large costs are
exceptional. For theoretical purposes, it is preferable to obtain
approximation ratios independent of the cost of any large item, which
may be as much as 1 (for which the corresponding approximation bound
is vacuous). In the sequential setting, the well-known technique of
``partial enumeration'' removes the dependence on the maximum cost and
obtains the same approximation ratio as the cardinality constraint
\citep{sviridenko}. In partial enumeration, one initializes the
solution ($x$ or $Q$) with different combinations of a constant number
of initial elements (3 suffices), hoping to guess the largest margin
items in the optimal solution.  It is easy to see that partial
enumeration extends here as well, and obviously can be done in
parallel without increasing the depth.

\begin{theorem}
  In $\bigO{\frac{\log n}{\eps^2}}$ depth, one can compute an
  $\apxless \parof{1 - e^{-1}}$-multiplicative approximation to
  maximizing a normalized monotone submodular function subject to a
  knapsack constraint.
\end{theorem}

Note that although partial enumeration does not increase the depth, it
does increase the total number of oracle queries and work by a
$\bigO{n^3}$-multiplicative factor.  \citet{en-17} recently obtained
an alternative to partial enumeration that increases the total work
and number of oracle queries by a
$\bigO{\exp{\poly{1/\eps}}}$-multiplicative factor instead, which is
preferable for modest values of $\eps$. The techniques may extend
here, but the details are tedious and beyond the scope of this paper.


\section{Concentration bounds}

\subsection{Online Chernoff inequalities}

We employ the following online extension of multiplicative Chernoff
inequalities, previously used in \cite{young-00,cq-18}.
\begin{lemma}
  \labellemma{online-chernoff}
  Let $X_1,\dots,X_n,Y_1,\dots,Y_n \in [0,1]$ be random variables and
  let $\eps > 0$ be sufficiently small.
  If
  \begin{math}
    \evof{X_i \given X_{1},\dots,X_{i-1},Y_1,\dots,Y_i} \leq Y_i
  \end{math}
  for $i \in [n]$, then for any $\delta > 0$,
  \begin{align*}
    \probof{\sum_{i=1}^n X_i \geq \epsmore \sum_{i=1}^n Y_i + \delta}
    \leq                        %
    \epsmore^{-\delta}.
  \end{align*}
\end{lemma}

\subsection{Concentration bounds for decay processes}

\begin{lemma}\labellemma{concentration-decay}
  Let $X_1,X_2,\dots \in \naturalnumbers$ where $X_1 = n$, and for $i
  \geq 2$,
  \begin{math}
    \evof{X_i \given X_1,\dots,X_{i-1}} \leq \max{\epsless X_{i-1},
      1}.
  \end{math}
  Then there exists a constant $c > 0$ such that, for
  $k \in \naturalnumbers$
  \begin{align*}
    \probof{X_k > 1} \leq \exp{- c \frac{k \eps}{\log n}}.
  \end{align*}
\end{lemma}

\begin{proof}
  Let $Y_i = \log{X_{i-1}} - \log{X_i}$.
  Then
  \begin{math}
    0 \leq Y_i \leq \log n,
  \end{math}
  and
  \begin{math}
    \sum Y_i \leq \log n.
  \end{math}
  For each $i$, if $X_{i-1} = 1$ (and the process has essentially
  halted), we have $Y_i = 0$. Otherwise,
  we have
  \begin{align*}
    \evof{Y_i \given X_1,\dots,X_{i-1}} %
    = %
    \log{X_{i-1}} - \evof{\log{X_i}}
    \tago{\geq} %
    \log{X_{i-1}} - \log{\evof{X_i}}
    = %
    \log{\frac{1}{1 - \eps}}    %
    \geq %
    c \eps
  \end{align*}
  for some constant $c > 0$.  The claim now follows from applying the
  following lemma to the variables $\setof{Y_i / c \eps}$ with
  $K = \frac{\log n}{c \eps}$.
\end{proof}

\begin{lemma}
  Let $Y_1,Y_2, \dots \geq 0$ and $K > 0$ such that
  \begin{enumerate}
  \item For any $\ell \in \naturalnumbers$, $\sum_{i=1}^{\ell} Y_i \leq K$.
  \item For any $\ell \in \naturalnumbers$, if
    $\sum_{i=1}^{\ell} Y_i < K$, then
    \begin{math}
      \evof{Y_{\ell+1} \given Y_1,\dots,Y_{\ell}} \geq 1.
    \end{math}
\item For any $\ell \in \naturalnumbers$, if
    $\sum_{i=1}^{\ell} Y_i = K$, the $Y_{\ell+1} = 0$.
  \end{enumerate}
  Then for $i \in \naturalnumbers$,
  \begin{math}
    \probof{\sum_{j=1}^{i \roundup{K}} Y_j < K} \leq e^{- c i},
  \end{math}
  for some absolute constant $c >0$.
\end{lemma}

\begin{proof}
  Define $Z_1,Z_2,\dots \geq 0$ by
  \begin{align*}
    Z_i =
    \begin{cases}
      Y_i &\text{if } \sum_{j < i} Y_i < K, \\
      1 &\text{if } \sum_{j < i} Y_i = K.
    \end{cases}
  \end{align*}
  For each $i$,
  \begin{math}
    \sum_{j \leq i} Y_j = K \iff \sum_{j \leq i} Z_i \geq K.
  \end{math}
  Moreover, we have $\evof{Z_i \given Z_1,\dots,Z_{i-1}} \geq 1$ for
  all $i$ and any values of $Z_1,\dots,Z_{i-1}$.

  For $L = \roundup{K}$, we divide the $Z_i$'s into groups of $L$. For
  $i \in \naturalnumbers$, let
  \begin{math}
    W_i = \sum_{j = 1}^{L} Z_{(i-1) L + j}.
  \end{math}
  For each $i$, we have
  \begin{math}
    0 \leq W_i \leq 2 K -1
  \end{math}
  unconditionally,
  and
  \begin{math}
    \evof{W_i \given W_1,\dots,W_{i-1}} \geq K
  \end{math}
  for any values of $W_1,\dots,W_{i-1}$.

  We now have
  \begin{align*}
    \probof{\sum_{j=1}^{i L - 1} Y_j < K} %
    =                                     %
    \probof{\sum_{j=1}^{i L - 1} Z_j < K}
    =                           %
    \probof{\sum_{j=1}^{i} W_j < K} %
    \leq                           %
    \probof{\sum_{j=1}^i W_j < 2 K} %
    \tago{\leq}                               %
    e^{- c i}
  \end{align*}
  by \tagr Chernoff inequalities, for some constant $c > 0$.
\end{proof}


\end{document}